\documentclass[10pt,journal,compsoc]{IEEEtran}  % Comment this line out
\usepackage{booktabs} % nicely formatted tables
\usepackage{url}
\usepackage{algorithmic}
\usepackage{algorithm}
\usepackage{physics}
\usepackage{soul,color,graphicx,float}
\usepackage{eurosym,booktabs,multirow}
\usepackage{mathtools}
\usepackage{graphicx}
\usepackage{subfigure}
\usepackage{amsmath,amssymb,amsfonts}
\usepackage{amsthm}
\usepackage[euler]{textgreek}
\definecolor{Orange}{rgb}{1,0.5,0}

\newtheorem{definition}{Definition} [section]
\newtheorem{theorem}{Theorem}  [section]

\usepackage{xspace}
\usepackage{ragged2e}
\usepackage{ifthen}
\newboolean{authnotes}
\setboolean{authnotes}{true}

\ifthenelse{\boolean{authnotes}}
{
\newcommand{\usman}[1]{\footnote{{\bf Usman: #1}}}
\newcommand{\aftab}[1]{\footnote{{\bf Aftab: #1}}}
\newcommand{\tony}[1]{\footnote{{\bf Tony: #1}}}
\usepackage{draftwatermark}
\SetWatermarkText{}
}
{
\newcommand{\usman}[1]{}
\newcommand{\aftab}[1]{}
\newcommand{\tony}[1]{}
\usepackage{draftwatermark}
\SetWatermarkText{}

}

\begin{document}

\title{PoFEL: Energy-efficient Consensus for Blockchain-based Hierarchical Federated Learning 

\author{Shengyang Li, Qin~Hu, Zhilin Wang}% <-this % stops a space
\thanks{This work was supported in part by the US NSF under Grant
CNS-2105004. (Corresponding Author: Qin Hu.)}
\IEEEcompsocitemizethanks{
\IEEEcompsocthanksitem Shengyang Li is with the Department of Electrical and Computer Engineering, Indiana University-Purdue University Indianapolis (IUPUI), IN, 46202, USA. This work was partially done when he was with the Department of Computer and Information Science at IUPUI.
E-mail: sl137@iu.edu

\IEEEcompsocthanksitem Qin Hu and Zhilin Wang are with the Department of Computer and Information Science, Indiana University-Purdue University Indianapolis (IUPUI), IN, 46202, USA. E-mail: \{qinhu, wangzhil\}@iu.edu
}
}
\IEEEpubidadjcol
\IEEEtitleabstractindextext{%
\justify
\begin{abstract}
Facilitated by mobile edge computing, client-edge-cloud hierarchical federated learning (HFL) enables communication-efficient model training in a widespread area but also incurs additional security and privacy challenges from intermediate model aggregations and remains the single point of failure issue. To tackle these challenges, we propose a blockchain-based HFL (BHFL) system that operates a permissioned blockchain among edge servers for model aggregation without the need for a centralized cloud server. The employment of blockchain, however, introduces additional overhead. To enable a compact and efficient workflow, we design a novel lightweight consensus algorithm, named Proof of Federated Edge Learning (PoFEL), to recycle the energy consumed for local model training. Specifically, the leader node is selected by evaluating the intermediate FEL models from all edge servers instead of other energy-wasting but meaningless calculations. This design thus improves the system efficiency compared with traditional BHFL frameworks.
To prevent model plagiarism and bribery voting during the consensus process, we propose Hash-based Commitment and Digital Signature (HCDS) and Bayesian Truth Serum-based Voting (BTSV) schemes. Finally, we devise an incentive mechanism to motivate continuous contributions from clients to the learning task. Experimental results demonstrate that our proposed BHFL system with the corresponding consensus protocol and incentive mechanism achieves effectiveness, low computational cost, and fairness.
\end{abstract}
\begin{IEEEkeywords}
Hierarchical federated learning, blockchain, consensus, incentive mechanism
\end{IEEEkeywords}
}

\maketitle
\section{Introduction}

%Mobile Edge Computing (MEC) has attracted significant academic and industry interests in recent years with the development of hardware technologies and increasing needs for computing resources. The core idea of MEC is that it achieves 
\IEEEPARstart{A}S the increasing amount of data generated at end devices causes massive communication and latency overhead for traditional cloud computing, mobile edge computing (MEC) \cite{ren2019survey} has been implemented to serve end users promptly at proximity to achieve cloud computing capabilities at the edge of the network. Machine learning (ML) tasks at end devices are offloaded by sending data to edge servers for analysis and calculation. %However, the data collected by the end devices are sensitive, and uploading all of the data to the edge servers might lead to severe privacy concerns.
Due to data privacy and communication cost concerns, 
%Among various machine learning frameworks, 
federated learning (FL) \cite{mcmahan2017communication,banabilah2022federated} is introduced in MEC to form a new distributed ML paradigm, namely federated edge learning (FEL), % to tackle the arising privacy concerns 
where the edge server acts as the parameter server and proximate devices within the server's cover range work as FL clients to train ML model collaboratively. %For MEC, this removes the need to collect private data from clients (end devices) to perform learning. Instead, clients perform local training and upload the model updates to the edge server without sharing their data. 
However, since end devices are widely spread across a vast area, one single server may not be capable of implementing FEL among all of them, stimulating the idea of hierarchical federated edge learning (HFL) \cite{liu2020client}. According to the HFL framework illustrated in Fig. \ref{fig_HFEL}, the cloud server at the top layer coordinates the training of the global model, multiple edge servers at the middle layer aggregate local models from end devices to generate intermediate models, and end devices at the bottom layer train the local models based on their local data.
\begin{figure}[h ]
\centering
\subfigure[Traditional HFL.]{
\label{fig_HFEL}
\includegraphics[width=0.23\textwidth]{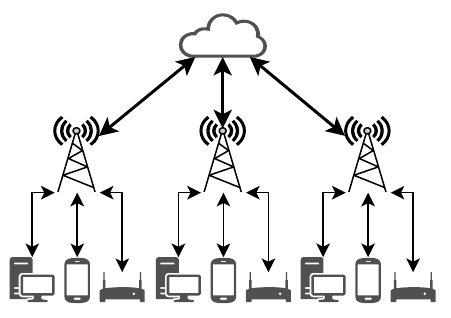}}
\subfigure[Blockchain-based HFL.]{
\label{fig_BCHEFL}
\includegraphics[width=0.23\textwidth]{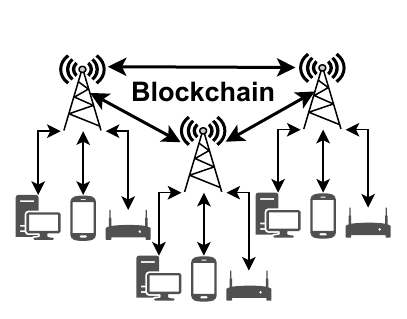}}
\caption{HFL frameworks.}
\label{diff_HFEL}
\end{figure}
% \begin{figure}[b]
% 	  %\vspace{-0.4cm}
%       \centering
%       \includegraphics[width=0.95\columnwidth]{HFEL.pdf}
%       %\vspace{0cm}
%       \caption{HFEL framework.}
%       \label{fig:HFEL}
% \end{figure}

However, the implementation of HFL still faces two major challenges: i) the single cloud server risks the whole system with the single-point failure; and ii) the layered topology introduces additional model security and privacy issues, such as model poisoning attacks by altering the intermediate models \cite{liu2021privacy} and leaking the local models of end devices at the edge servers. %  Model poisoning attacks refer to the attacks that leads to inaccurate or biased global model and local model update leakages could be edge servers leaking local models of end devices to others that violate privacy.
% \textbf{should the trust/accountability issue target the security and privacy issues? such as the faulty local models from end devices, faulty aggregation of local models at edge servers, and edge servers leaking the local models of devices to others.}
% multiple edge servers that work for their interests as individual entities but would like to collaborate to accomplish a learning task require a trustworthy and reliable platform. 
Such inherent issues become an obstacle to deploying HFL in practice. 
%As a result, edge servers require a trustworthy and tamper-proof FL process in which each edge server will participate instead of letting a central server manage FL. 

One promising solution is introducing the blockchain \cite{nakamoto2008bitcoin,huo2022comprehensive} into HFL to omit the need for a central server, as well as provide transparent and reliable model exchanges and global model aggregations, so as to realize better security performance. As shown in Fig. \ref{fig_BCHEFL}, we propose a blockchain-based HFL (BHFL),  where edge servers of HFL form a consortium blockchain network and exchange intermediate models on the blockchain. The global model aggregation is then performed in a distributed and traceable manner instead of the centralized way in the traditional HFL. The consortium blockchain provides a transparent and secure platform for FL, and all edge servers can track and audit the whole FL process independently. The lower layer of our proposed framework is the same as that of the traditional HFL, where each edge server and the connected clients 
%an FEL process between the edge server and connected clients. Edge server and its clients 
exchange local model updates and aggregated intermediate models several times in multiple FEL iterations. 
Limited related work about implementing blockchain in HFL can be found in \cite{sarhan2022hbfl,zhang2021bc,huang2023distance} while others involving the keywords of blockchain, hierarchical, and FL are applying hierarchical blockchain to FEL \cite{fu2023incentive,chai2020hierarchical, %} and hierarchical FL \cite{
xu2022mudfl}. 
%\textbf{\textit{add those seemingly related work and summarize their feature to indicate their irrelevance to our work.}} 
In \cite{sarhan2022hbfl}, the authors replace the cloud server with a dedicated blockchain as the top layer, and the authors in \cite{zhang2021bc,huang2023distance} append a blockchain between edge servers and the cloud server to work as a storage and record platform. Different from these existing studies deploying the blockchain as an additional part of HFL, our proposed system fully integrates blockchain with HFL through deploying blockchain on edge servers, thus being more compact and efficient compared with these existing BHFL frameworks.

% why do we need a new consensus algorithm? why the existing ones not working in our framework?
Nevertheless, it is nontrivial to deploy the proposed BHFL framework. First, the existing consensus protocols cannot be arbitrarily implemented in this BHFL system since they are usually working as an additional process to determine the generator of a new block besides finishing the FL computation task, bringing extra computational overhead to edge servers and slowing down the training speed of the global model. %Second, when the leader node election which determines the edge server (node) who has the rights to publish and broadcast the new block, considers the votes coming from all the participating edge servers, malicious edge servers could bribe other edge servers to cast votes that affects the votes tallying process to fail or manipulate the leader node election process. 
Second, the intermediate models of edge servers need to be uploaded to the blockchain during the consensus process for verification and block generation, which makes it easy for some malicious edge servers to plagiarize models shared by other edge servers without performing the actual FEL process. 
Moreover, given that numerous end devices conduct FL tasks via the coordination of distributed edge servers, consuming extensive computing and communication resources, it is imperative to provide sufficient incentives to them to compensate for their cost and motivate their long-term participation.

%the lack of an energy-efficient consensus algorithm at the consortium blockchain can make the model aggregation 
%aggregating intermediate models at edge servers to the global model in a distributed manner. Since FL process requires a significant amount of computing resources, the process done in the blockchain should not introduce further significant computational overhead. Moreover, there is a plagiarism concern during intermediate model exchanges among edge servers for the consensus process, where one edge server plagiarizes models shared from other edge servers without performing FEL process.
%Rewards distributions among edge servers and model owner (task publisher) to provide constant motivation for FEL processes. 

In this paper, we design an energy-efficient consensus algorithm, named PoFEL, to address the first challenge, where the intermediate models of FEL at edge servers are utilized to determine the leader in the consortium blockchain. To solve the plagiarism issue, we design a cryptographic scheme named HCDS, combining Hash, Commitment Scheme, and Digital Signature, to validate the process of model exchanges among edge servers without worrying about their dishonesty. We employ cosine similarity to measure the difference between the aggregated global model and each intermediate model from the edge server, %. The scheme is applied individually at each edge server, and 
based on which each edge server votes the edge server with the lowest difference to be the leader node. All votes will be sent to a smart contract for vote tallying, 
with a weighted voting system that adopted Bayesian Truth Serum (BTS) scoring method to tackle the potential issue of bribery voting. %This process only involves simple model aggregation and cosine similarity computations, which is energy-efficient. The leader selected from edge servers is responsible for global model aggregation in each iteration. Then the global model is stored in the block and broadcast to other edge servers. To support multiple learning tasks working concurrently in the system, we also design a simple first-in-first-choose queue that stores time information of all learning tasks. The consensus algorithm will iterate the queue to select the learning task for cosine similarity calculation in each iteration. 
 %\textbf{\textit{In addition, we design a simple first-in-first-choose queue for the consensus algorithm as the criteria of leader node election to achieve consensus among multiple FL tasks concurrently.}} 
Finally, we propose a Stackelberg game-based incentive mechanism to address the motivation challenge.

Overall, the main contributions of this work can be summarized as follow:
\begin{itemize}
    \item We propose a BHFL framework with edge servers running the consortium blockchain for model aggregation and distributed end devices capable of collecting diverse data samples to benefit model training. %, where multiple FL tasks can be performed simultaneously at this platform to provide more efficient and powerful FL services with various training data.
    \item To optimize the energy efficiency of the proposed BHFL system, we design a novel consensus algorithm, namely Proof of Federated Edge Learning (PoFEL), to reduce the extra computational cost for reaching blockchain consensus using the edge servers' intermediate models as an evaluation metric for determining the leader node. We design a Hash-based Commitment and Digital Signature (HCDS) scheme to prevent model plagiarism in the consensus process involving model exchanges. Bayesian Truth Serum-based Voting (BTSV) is also developed to minimize the chance of bribery voting.  %A simple iterative First-in-First-choose mechanism in the consensus algorithm is introduced to support concurrent multiple learning tasks
    \item To ensure the motivation of participants involved in finishing learning tasks, we proposed a two-stage Stackelberg game-based incentive mechanism to help model owners decide the rewards being distributed to each FEL cluster and facilitate the decision of edge server as the FEL coordinator regarding the optimal amount of computational resources to be spent.% in each FL task. The rewards allocated to each FEL cluster will be further distributed to end devices by their associated edge server.
    \item We evaluate our proposed mechanisms through extensive experiments. The experiment results prove that our mechanisms achieve low computation cost, fairness, and effectiveness.

\end{itemize}

The remainder of this paper is organized as follows. Related work is discussed in Section \ref{sec:related}. The system workflows and adversary models are illustrated in Section \ref{sec:sysmod}. The consensus algorithm PoFEL is presented with details in Section \ref{sec:consensus}, and an incentive mechanism based on a two-stage Stackelberg game is introduced in Section \ref{sec:incent}. Security analysis of the proposed system is given in Section \ref{sec:security} and experiment evaluations are presented in Section \ref{sec:experiment}. Finally,  we conclude the paper in Section \ref{sec:discussion}.
  
\section{Related Work}
\label{sec:related}
\subsection{Hierarchical Federated Edge Learning}
Traditional FL studies rely on a central server (e.g., cloud) to collect local models from clients. A lot of recent research has been proposed to introduce edge computing \cite{peng2018survey} into traditional FL and formulate various FEL systems \cite{abdellatif2022communication,liu2020client,lim2021decentralized,chen2023enhanced}. In \cite{liu2020client}, the authors propose a client-edge-cloud HFL system where edge servers are deployed as intermediate model aggregators to reduce the overall communication cost and the processing burden of the cloud, along with the design of
%. The factor inspiring this paper is that the communication latency and cost from the client to the edge server are generally smaller than that in the case of client-cloud communication. This also relieves the heavy burden of the cloud server directly processing received model updates since the edge server performs intermediate aggregations before the global aggregations at the cloud. The authors then invent 
a hierarchical FL aggregation algorithm named HierFAVG and its convergence analysis. % is designed to define the aggregation process of the proposed system framework. They further provide the convergence analysis to validate the effectiveness of HierFAVG given non-IID user data. 
This paper serves as a very promising and classical FEL architecture. Various following papers investigated the remaining issues of FEL systems. The authors in \cite{abdellatif2022communication} devise a novel client-edge assignment scheme based on client location, data distribution, and communication constraints under the same FEL system architecture, where the optimization goal is to minimize FL convergence time, as well as communication and computation cost of the whole system. Similarly, Lim et al. \cite{lim2021decentralized} extend the FEL system proposed in \cite{liu2020client} to support multiple model owners, % with the problem of joint resource allocation being resolved, 
where each edge server serves as a cluster head that connects with a model owner. %The authors propose a joint resource allocation for their system framework to minimize the cost.
The authors in \cite{chen2023enhanced} investigate the model training challenges in a Hybrid Hierarchical Federated Edge Learning (HHFEL) system with device stragglers. They design an online semi-asynchronous FL approach to improve training efficiency by optimizing resource allocation and device selection. %This method reduces device stragglers and encourages efficient local model training.
%\qh{The authors in \cite{chen2023enhanced} investigate the model training challenges in a Hybrid Hierarchical Federated Edge Learning (HHFEL) system, focusing on the training efficiency degradation introduced by device stragglers in heterogeneous networks. The problem is modeled as an online Markov Decision Process and an online semi-asynchronous FL process is designed between the edge layer and device layer to schedule the resource allocation between edge and device and device selection. (too long to understand, too many repeated words.)}

\subsection{Blockchain-based Federated Learning}
Some studies about blockchain-based federated learning (BCFL) systems directly replace the central server with a blockchain network involving a specified set of miners \cite{kim2019blockchained, Ma2020WhenFL, wang2023incentive}. 
%In \cite{kim2019blockchained}, the authors propose a BCFL system to solve the single point failure introduced by the single central server and the motivation of clients constantly providing model updates to the model owner. The blockchain's consensus algorithm and incentive mechanism inherently address these two issues. 
Other papers \cite{hua2020blockchain,weng2019deepchain} with alternative approaches employ FL clients as blockchain nodes without introducing additional miners. Since mining is usually computationally intensive for an FL client in addition to local model training, %computation power becomes a significant factor to clients are the best candidates for being miners in the blockchain.
%Therefore in this line of research,  
partial or all clients with enough computation power may work as blockchain miners for validating model updates and new block generation. 
%\textbf{\textit{However, mining is usually computationally heavy for a client in addition to local model training. Therefore, powerful equipment is the best candidate for being a miner in the blockchain. In \cite{nguyen2021federated}, the authors propose a framework called FLchain, in which the blockchain consists of edge servers and a portion of local devices. }}%All edge servers serve as miners, and some powerful local devices which are already training nodes also participate in blockchain mining as blockchain nodes. }}

%\textbf{\textit{PLEASE rewrite the whole paragraph according to our discussion on Wechat! }}
The idea of blockchain could also be applied in FEL to overcome single-point failure and other security challenges including using blockchain \cite{jin2023lightweight,liu2023blockchain,nguyen2021federated} or hierarchical blockchain \cite{fu2023incentive,chai2020hierarchical,xu2022mudfl}. Both works in \cite{jin2023lightweight,liu2023blockchain} propose similar BCFL systems for the scenario of FEL where the edge servers are responsible for aggregating local models from the local devices. The edge servers also formulate a blockchain that is responsible for consensus and auditing the local models. 
In \cite{nguyen2021federated}, the authors propose a framework called FLchain, in which the blockchain is maintained by edge servers and a portion of powerful local devices. 
Chai et al. \cite{chai2020hierarchical} introduce a hierarchical blockchain structure into FEL to achieve knowledge sharing among the internet of vehicles. Roadside units (RSUs) are grouped according to their locations and communication ranges, and then a ground blockchain is formulated by a group of RSUs and used for recording model updates submitted from low-level FL clients, i.e., vehicles. %\textbf{\textit{Multiple ground blockchains are formulated among roadside units (RSUs).}} 
%FL clients upload model updates to nearby RSUs so that 
RSUs would integrate these model updates with their own trained model updates and then upload integrated results to the top blockchain composed of base stations. 
%and the middle layers of edge servers form a blockchain and propose a lightweight consensus algorithm called Proof-of-Knowledge that selects the node with the most accurate learning result as the leader. % to support the proposed framework. 
Xu et al. \cite{xu2022mudfl} propose a similar hierarchical blockchain-enabled FL framework in which an inter-microchain network interconnects multiple independent microchain consensus networks composed of clients and intermediate aggregators. The FL server and aggregators form the inter-microchain and a subset of them are selected as a validator committee to execute Byzantine Fault Tolerance (BFT) consensus to record microchains' checkpoints. Both studies apply the hierarchical blockchain to FL but the extra system cost in terms of latency and energy is not discussed.

\subsection{Blockchain-based Hierarchical Federated Learning}
%\textbf{\textit{PLEASE rewrite the whole paragraph according to our discussion on Wechat! }}
HFL proposed in \cite{liu2020client} inspires several studies that employ blockchain in HFL. In \cite{sarhan2022hbfl}, the authors propose a collaborative IoT intrusion detection system among organizations using a blockchain-based hierarchical FL framework. % where the design of FL system follows HFL structure. 
Given each organization possesses multiple IoT endpoints, organizations work as middle-layer servers and aggregate their collected ML-based intrusion detection models from IoT endpoints. Resulting models are further uploaded to a permissioned blockchain that works as the cloud server for global model aggregation.  Zhang et al. \cite{zhang2021bc} implement a data transmission model using the idea of edge computing, FL, and blockchain. % where blockchain is deployed between edge servers and cloud server as verifier and recorder. %, where additional edge servers are introduced between the cloud server and clients. 
Edge servers collect local model updates from low-level end devices and transmit them to the blockchain for verification and aggregation. Then the cloud server retrieves aggregated models from the blockchain and further aggregates these intermediate models into the global model. Huang et al. \cite{huang2023distance} propose a multi-layer blockchain-based HFL network consisting of an IoT devices layer, edge layer, blockchain layer, and cloud layer. After analyzing the relationships between training model error, IoT device associations, and data distribution, they propose distance-aware hierarchical federated learning (DAHFL) to %optimize IoT device associations using the data distribution of each IoT device to 
minimize training model error. 
%\qh{A theoretical analysis of the bottleneck in model accuracy caused by the imbalanced data distribution is given. Then the relations between the model error, IoT device associations, and data distribution are modeled. Finally, distance-aware hierarchical federated learning (DAHFL) is proposed to reduce model error by optimizing IoT device associations based on the data distribution of each IoT device. (too long. can you rephrase these sentences into one or two sentences?)}
In all the above-mentioned studies, blockchain is deployed as an additional part of the HFL structure, without considering the extra system overhead.

%\textbf{\textit{PLEASE replace the following two paragraphs with one paragraph summarizing the difference (or shortage) of the existing work, mainly [3],[4], from  our work.}}
%In conclusion, the existing studies about blockchain-based HFL are limited. 
%limited studies have addressed the single-point failure, as well as the security and privacy threats in HFL using blockchain.  
Different from the existing BHFL studies \cite{sarhan2022hbfl,zhang2021bc}, our proposed framework is designed to be energy-efficient by empowering edge servers to form the blockchain system without employing additional miners. To tackle the challenges brought by the blockchain, %such as the motivation of clients, model plagiarism, and energy cost, 
we design a lightweight consensus algorithm utilizing intermediate model aggregation results for leader election, which is enhanced with plagiarism- and bribery-free schemes, and a Stackelberg game-based incentive mechanism to encourage the participation of clients.

\section{System Overview}
\label{sec:sysmod}

\begin{figure*}[ht]%
   \centering
      \includegraphics[width=1.9\columnwidth]{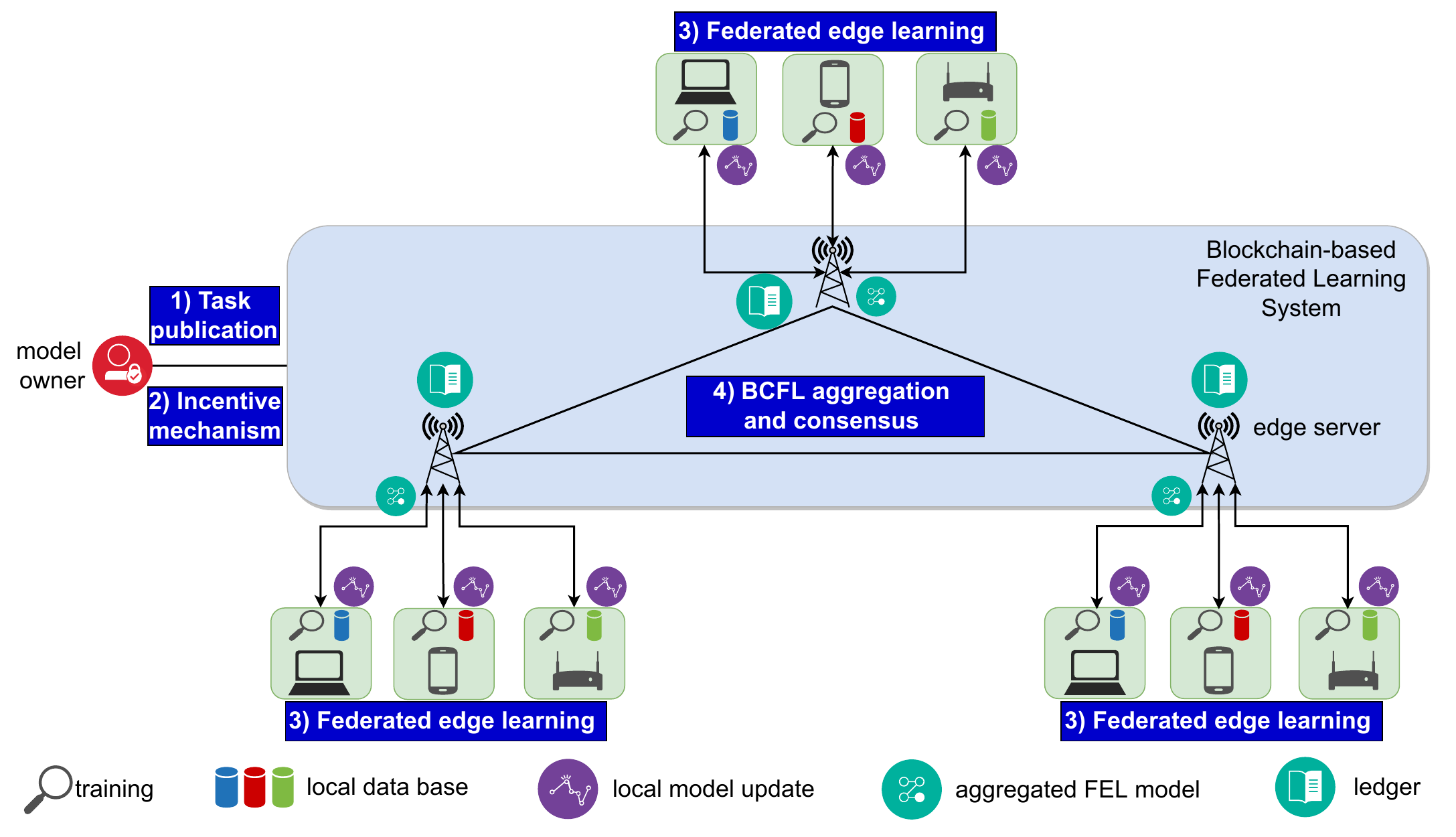}
      %\vspace{0cm}
      \caption{The topology of our proposed system.}
      \label{fig:modelstructure}
\end{figure*}

% \begin{figure}[t]
% 	  %\vspace{-0.4cm}
%       \centering
%       \includegraphics[width=1\columnwidth]{modelstructure.pdf}
%       %\vspace{0cm}
%       \caption{System overview}
%       \label{fig:modelstructure}
% \end{figure}
%\subsection{System Overview}
As illustrated in Fig. \ref{fig:modelstructure}, our proposed BHFL system has two layers: i) the upper layer is the blockchain-based federated learning (BCFL) system running on edge servers\footnote{For simplicity, we refer to edge servers as BCFL nodes.}, and ii) the bottom layer includes all FEL systems with each consisting of one BCFL node and multiple end devices (i.e., clients). We refer to a BCFL node as \textit{$e_i$} and $E = \{e_1,\cdots,e_i,\cdots,e_N\}$ represents all BCFL nodes in the system, where $N$ is the total number of BCFL nodes. 

\subsection{System Workflow}
Specifically, there are four main procedures in the workflow of the proposed system:

1) \textbf{Task Publication:} A user sends a learning task request to the BCFL network so that all BCFL nodes can receive the information and apply an evaluation scheme to determine whether they accept the task. %This evaluation scheme could vary from individual to individual and be decided by the FEL network. For example, clients in an FEL network will vote on whether to join the task based on whether they have sufficient qualified data and computing resources for the learning task. Since clients within one FEL network could possess various and distinct kinds of data, given a learning task being accepted by an FEL network, this implies there could be only a fraction of all clients in an FEL network working on the same learning task. Meanwhile, other clients are idle or working on different learning tasks published by other requestors. Besides, 
The task information includes the user's identity and learning task description, which  % when the task is published, and other requirements, 
will also be recorded on the blockchain. 

2) \textbf{Incentive Mechanism:} A simple two-stage Stackelberg game-based \cite{simaan1973stackelberg} incentive mechanism is applied between the task publisher (i.e., model owner) and each participating BCFL node before beginning the following FEL procedure to determine rewards being assigned to each BCFL node and the amount of computational resource each BCFL node invests. These rewards are the payment of the FL task for an FEL cluster consisting of the BCFL node and associated end devices as compensation for training local models using their own data. 

3) \textbf{Federated Edge Learning:} The BCFL node will distribute the learning task or updated global model to connected clients in each participating FEL cluster. Afterward, clients will prepare the training data, perform the local training and upload their model updates to their connected BCFL node. The BCFL node will aggregate received local model updates into an FEL model which will be returned to clients to begin another iteration\footnote{The aggregation method used in each FEL system could be different and determined by the FEL system itself. Some popular and well-known methods include federated stochastic gradient descent (FedSGD) and federated averaging (FedAvg) \cite{mcmahan2017communication}.}.

4) \textbf{Global Model Aggregation and Consensus in BCFL:} After a specific period of FEL, the BCFL system will generate a block, containing the latest FEL models from all BCFL nodes and the resulting updated global model, and broadcast it to be recorded on the blockchain. This process is taken place at a constant rate in a round-by-round manner, and every BCFL node will participate in this process. In each round, all BCFL nodes perform the consensus mechanism to determine a leader for block generation. %It is worth noting that some of these latest FEL models could have reached convergence. Furthermore, there could be FEL networks that are not working on any learning tasks. For instance, an FEL network finishes all its previous learning tasks and awaits appropriate learning tasks to join. However, BCFL nodes in these idle FEL networks are still required to participate in the BCFL block generation iteration processes even if they do not have FEL models to exchange with other nodes. They will receive all FEL models from other BCFL nodes to perform the consensus algorithm. 
 Each BCFL node first broadcasts its FEL model and then computes the updated global model using all received FEL models. The consensus mechanism is based on the similarity between each FEL model and the updated global model. %Since multiple learning tasks are running in the system, the consensus algorithm is applied based on one single learning task in each BCFL block generation iteration. This single learning task is chosen based on a simple time rule. 
 And the BCFL node submitting the FEL model with the highest similarity will become the leader, which will be responsible for the generation of a new block %aggregating all FEL models from the remaining learning tasks to updated global models. Then it will generate the block ledger 
storing important information, such as the leader node identity, all submitted FEL models, and the updated global model. %Finally, it will broadcast the ledger to all other BCFL nodes. 
%Furthermore, the leader BCFL node will get a fixed payment for generating the block. These rewards come from the corresponding learning task publisher's budget. 
Once all BCFL nodes have received and verified the new block through broadcasting, they will add it to their own ledger to finish the blockchain updating.
%proceed to the incentive mechanism. % or termination procedure. 
More details about the consensus algorithm will be provided in Section~\ref{sec:consensus}.

The learning task will be terminated once the loss function of the global model is minimized or the user-specified completion time is expired.

\subsection{Adversary Models}
\label{sec:adversary}
In our proposed BHFL system, there exist potential adversaries impeding the implementation of the above four main procedures, especially during the consensus process. In the below, we define the assumptions, goals, and capabilities of malicious BCFL nodes as adversary models, including model plagiarism and bribery voting.
\subsubsection{Model Plagiarism}
In the context of consensus algorithm in BCFL, all BCFL nodes exchange their FEL models initially to facilitate model aggregation, leader election, and verification. However, due to the distributed manner of the network, achieving simultaneous model exchanges across all nodes presents a challenge. Certain nodes may receive models earlier than others, creating an opportunity for malicious nodes to engage in model plagiarism. These nodes can exploit this time discrepancy to fabricate models as their own by plagiarizing other received models, enabling them to obtain rewards without utilizing their own computing resources.

\textbf{Assumption:} Adversaries can include any BCFL node within the BCFL system, without any additional resource or equipment requirements. The adversaries decide to cease the FEL process to save their own computing resources.

\textbf{Goals:} After receiving FEL models from other benign BCFL nodes, an adversary utilizes one or more of the received models to fabricate a new FEL model and present it as its own learned FEL model. This can be done by direct copying or editing a received model or merging several received models. Subsequently, the adversary broadcasts the plagiarized model as expected. Since there is no criteria to distinguish the plagiarized model, by following this approach, the adversary can acquire rewards for the plagiarized model for free without being recognized by other benign BCFL nodes.

\textbf{Capabilities:} For an adversary to benefit from the plagiarized FEL model and obtain rewards, they must possess adequate network conditions to timely broadcast the plagiarized models before the model aggregation occurs. Since the adversaries can be any BCFL node, it is capable of receiving, copying or editing, and broadcasting the FEL models. 

\subsubsection{Bribery Voting}\label{subsec:bribery}
Within the proposed consensus algorithm, the process of determining the leader node involves vote tallying. However, in this scenario, a malicious node has the potential to incentivize other BCFL nodes to cast dishonest votes, thus disrupting the accuracy of the tallying results.

\textbf{Assumption:} For adversaries to successfully bribe other BCFL nodes, they must either have access to ample rewards or ensure that the act of bribery is financially beneficial (where the earned rewards exceed the cost of bribery). Additionally, the proportion of corruptible BCFL nodes is assumed to be lower than 50\%.

\textbf{Goals:} The adversaries gain unauthorized profits by engaging in bribery with other BCFL nodes, influencing them to vote in their favor and increasing their chances of being selected as leader nodes. Furthermore, the adversaries may choose to bribe other BCFL nodes to cast dishonest votes (random votes), thereby disrupting the integrity of the voting process.

\textbf{Capabilities:} The adversaries possess the ability to compromise a BCFL node, resulting in its willingness to faithfully cast votes in favor of the adversaries. Furthermore, they have covert means to distribute bribery rewards to other BCFL nodes without detection.

\section{PoFEL: Energy Efficient Consensus Algorithm Design}
\label{sec:consensus}
In this section, we will first discuss our proposed consensus mechanism PoFEL, followed by the detailed design of three main sub-procedures preventing adversarial attacks defined in Section \ref{sec:adversary}. 
%\textbf{Specifically, in section \ref{sec:HCDS} and \ref{sec:BTS} we design HCDS and BTSV to prevent attacks from the adversaries of model plagiarism and bribery voting seperately.} 

The overall process of our proposed consensus mechanism, named \textit{Proof of Federated Edge Learning (PoFEL)}, is summarized in Algorithm~\ref{al_1}. 
%We define $E=\{e_{1},e_{2},...,e_{N}\}$ as the set of all BCFL nodes, where $N$ is total number of BCFL nodes. 
%The BCFL system will apply the consensus mechanism every certain period of time. We also refer to the process where each time consensus mechanism is applied as a new block generation round. Assuming all BCFL nodes participate in the learning task, 
Denote $W(k)=\{w^{{1}}(k),\cdots,w^{{i}}(k),\cdots,w^{{N}}(k)\}$ as the set of FEL models, where $w^i(k)$ is the aggregated model of node $e_i$ at BCFL round $k$, and $N$ is the total number of BCFL nodes. %At new block generation round $k$, first, each 
In each round, every BCFL node $e_i$ runs the \textit{Hash-based Commitment and Digital Signature (HCDS)} function with its FEL model $w^i(k)$ as input (Line 2), which will lead to the broadcast of all FEL models and the verification of all received FEL models. 
%Then, we use $ts$ to denote the selected learning task being used as the criteria of our proposed \textit{Proof of Similarity (ME)}, and  $W^{ts}(k)=\{w_1^{ts}(k),w_2^{ts}(k),...,w_{n_{ts}}^{ts}(k)\}$ to represent the set of FEL models working on learning task $ts$ at new block generation round k, where $n_{ts}$ is the total number of BCFL nodes (FEL models) working on task $ts$. 
Then all BCFL nodes conduct \textit{Model Evaluation (ME)} to derive the updated global model $gw(k)$ from $W(k)$ (Line 3). Next, each BCFL node $e_i$ computes cosine similarities between each FEL model $w^i(k) \in W(k)$ and the updated global model $gw(k)$. Based on these similarities, $e_i$ will cast a vote $e^{i}_{best}(k)$ for the BCFL node with the highest similarity, along with the prediction of the voting result, denoted by $P^{i}(k)$ (Line 4).
%In addition to the vote, each BCFL node $e_i$ will also 
%node $e_i$ is required to predict the percentage of votes each node will receive, denoted by
%proportions of voters who will endorse each of the $N$ BCFL nodes as the leader node. 
%These predictions form a prediction set 
Specifically, $P^{i}(k)=\{p_1^{i}(k),\cdots, p_{N}^{i}(k)\}$ indicates the percentage of votes each node will receive. After gathering submissions from all nodes into $E_{best}(k)=\{e^{1}_{best}(k), e^{2}_{best}(k), \cdots, e^{N}_{best}(k) \}$ and $P(k)=\{P^{1}(k), P^{2}(k), \cdots, P^{N}(k)\}$, the smart contract for vote tally will employ the \textit{Bayesian Truth Serum-based Voting (BTSV)} scheme to generate a score that represents the honesty for each BCFL node. Such a scheme utilizes Bayesian Truth Serum (BTS) to evaluate how common each BCFL node is voted compared to the common predictions in $P(k)$. The cumulative historical score is then used to calculate the voting weight of each BCFL node and tally $E_{best}(k)$ based on these weights so that the leader node $e^*(k)$ at round $k$ can be determined (Line 5). 
%The leader node will be responsible for using aggregation function \textit{AG} to aggregate FEL models of the remaining learning tasks, and output updated global models set $GW(k)=\{gw^{t_1}(k), gw^{t_2}(k), ..., gw^{t_{n_t}}(k)\}$ where $n_t$ is the total number of learning tasks operating in the BCFL system (Line 6 in Algorithm~\ref{al_1}). 
Eventually, the leader node will package all important information into the \textit{New Block} and broadcast it to the whole BCFL system (Lines 6-7). 
 
\begin{algorithm}
\caption{Overall Process of PoFEL } %算法的名字
\label{al_1}
%\hspace*{0.02in} {\bf Input:} %算法的输入， \hspace*{0.02in}用来控制位置，同时利用 \\ 进行换行
%the set of local model updates in round $k$: $f(w)$\\
%\hspace*{0.02in} {\bf Output:} %算法的结果输出
%updated global model contained in new block: $G(w)$ 

\begin{algorithmic}[1]
\REQUIRE $W(k)$

\WHILE{round $k$}
\STATE run HCDS ($w^i(k)$) at every $e_i$
\STATE $(e^{i}_{best}(k)$, $P^{i}(k)$, $gw(k)) \gets$ ME ($W(k)$) at every $e_i$

\STATE each $e_i$\ submits $e^{i}_{best}(k)$ and $P^{i}(k)$ to smart contract

\STATE $e^*(k) \gets$ \uppercase{BTSV} $(E_{best}(k),P(k))$

%\STATE $GW(k) \gets$ AG($W(k)$) at $e^*(k)$ 
% 
\STATE $New\ Block\ \gets$ $e^*(k), W(k), gw(k)$

\STATE $e^*(k)$ broadcasts $New\ Block$

\ENDWHILE

\end{algorithmic}
\end{algorithm}

% \begin{algorithm}
% \caption{HCDS}\label{al_2}
% %\hspace*{0.02in} {\bf Input:} %算法的输入， \hspace*{0.02in}用来控制位置，同时利用 \\ 进行换行
% %the set of local model updates in round $k$: $f(w)$\\
% %\hspace*{0.02in} {\bf Output:} %算法的结果输出
% %updated global model contained in new block: $G(w)$ 

% \begin{algorithmic}[1]
% \REQUIRE $w^{i}(k)$

% \STATE $(R^{i}(k),D^{i}(k), TAG^{i}(k)) \gets$ HCDS($w^{i}(k)$) at $e_i$

% % TAG^{e_i}(k)=DS_{S_{e_i}}(D^{e_i}(k))
% \STATE $e_i$\ broadcasts\ $D^{i}(k)$ and $TAG^{i}(k)$ 
% \STATE Commit(${D(k,i)}$, ${TAG(k,i)}$) at $e_i$ 

% \STATE $e_i$\ broadcasts\ $R^{i}(k)$, $W^{i}(k)$and $TAG^{i}(k)$
% \STATE Reveal(${R(k,i)}$, ${W(k,i)}$, ${TAG(k,i)}$) at $e_i$ 

% \end{algorithmic}
% \end{algorithm}
\begin{algorithm}
\caption{HCDS}\label{al_2}
%\hspace*{0.02in} {\bf Input:} %算法的输入， \hspace*{0.02in}用来控制位置，同时利用 \\ 进行换行
%the set of local model updates in round $k$: $f(w)$\\
%\hspace*{0.02in} {\bf Output:} %算法的结果输出
%updated global model contained in new block: $G(w)$ 

\begin{algorithmic}[1]
\REQUIRE $w^{i}(k)$
\STATE \textbf{Initialize}: $r^{{i}}(k)$
%\STATE $(R^{i}(k),D^{i}(k), TAG^{i}(k)) \gets$ HCDS($w^{i}(k)$) at $e_i$
\STATE $d^{i}(k) = H(r^{i}(k) \ ||\  w^{i}(k))$
\STATE $tag^{i}(k) = DSign(d^{i}(k), SK_{i})$
% TAG^{e_i}(k)=DS_{S_{e_i}}(D^{e_i}(k))
\STATE broadcast\ $d^{i}(k)$ and $tag^{i}(k)$ 
\STATE receive\ ${D(k,i)}$ and $TAG(k,i)$
%\STATE Commit(${D(k,i)}$, ${TAG(k,i)}$) at $e_i$ 
\FOR{$l = 1$ to $N$ and $l \neq i$}
    
    \IF{$DVerify(tag^{l}(k), PK_{l}, d^{l}(k))$ = Accepted}
        \STATE Continue
    % \Else
    %     \While{$DS_{S_{e_i}}(tag_{m_{digest}^{e_i}})$ $\neq$ $m_{digest}^{e_i}$}
    %     \State $e_i$ rebroadcast $m_{digest}^{e_i}$
    %     \EndWhile
    \ENDIF
    
\ENDFOR
\STATE broadcast\ $r^{i}(k)$, $w^{i}(k)$, and $tag^{i}(k)$
\STATE receive\ $R(k,i)$, $W(k,i)$, and $TAG(k,i)$
%\STATE Reveal(${R(k,i)}$, ${W(k,i)}$, ${TAG(k,i)}$) at $e_i$ 
\FOR{$l = 1$ to $N$ and $l \neq i$}
    
    \IF{$H(r^{l}(k) \ ||\  w^{l}(k))$ = $d^{l}(k)$}
    \IF{$DVerify(tag^{l}(k), PK_{l}, H(r^{l}(k) \ ||\  w^{l}(k)))$ = Accepted}
    
    \STATE Continue
    \ENDIF
    \ENDIF
    
\ENDFOR
\end{algorithmic}
\end{algorithm}

% \begin{algorithm}
% \caption{HCDS}\label{al_3}
% \begin{algorithmic}[1]
% \REQUIRE $W^{i}(k)=\{w_1^{{i}}(k),w_2^{{i}}(k),...,w_{n_{i}}^{{i}}(k)\}$
% \STATE \textbf{Initialize}: $R^{i}(k)$$\enspace=\enspace$$\{ r_1^{{i}}(k),r_2^{{i}}(k),...,r_{n_{i}}^{{i}}(k) \}$,\\$D^{i}(k)$$\enspace=\enspace$$\{d_1^{{i}}(k),d_2^{{i}}(k),...,d_{n_{i}}^{{i}}(k)\}$,\\$TAG^{i}(k)$$\enspace=\enspace$$\{tag_1^{{i}}(k),tag_2^{{i}}(k),...,tag_{n_{i}}^{{i}}(k)\}$
% \FOR{j in $n_i$}
% \STATE $d^{i}_j(k) = H(r^{i}_j(k) \ ||\  w^{i}_j(k))$
% \STATE $tag^{i}_j(k) = DSign(d^{i}_j(k), SK_{e_i})$
% \ENDFOR
% \RETURN $R^{i}(k),D^{i}(k), TAG^{i}(k)$
% \end{algorithmic}
% \end{algorithm}
% \begin{algorithm}
% \caption{HCDS}\label{al_3}
% \begin{algorithmic}[1]
% %\REQUIRE $W^{i}(k)=\{w_1^{{i}}(k),w_2^{{i}}(k),...,w_{n_{i}}^{{i}}(k)\}$
% \STATE \textbf{Initialize}: $r^{{i}}(k)$, $d^{{i}}(k)$, $tag^{{i}}(k)$

% \STATE $d^{i}(k) = H(r^{i}(k) \ ||\  w^{i}(k))$
% \STATE $tag^{i}(k) = DSign(d^{i}(k), SK_{e_i})$

% \RETURN $r^{i}(k),d^{i}(k), tag^{i}(k)$
% \end{algorithmic}
% \end{algorithm}

\subsection{Hash-based Commitment and Digital Signature}
\label{sec:HCDS}
In our proposed PoFEL consensus algorithm, the broadcast of FEL models among BCFL nodes can lead to the issues of model plagiarism, where a malicious node might copy another node's submission with or without editing, or even fabricate an FEL model via merging multiple models from others and then claims the plagiarized model. %We will provide more details in the security analysis in section \ref{sec:security}.
To overcome this challenge, we devise a Hash-based Commitment scheme integrated with Digital Signature (HCDS) that enforces each BCFL node to strictly follow the BHFL protocol and submit their own FEL models.
%commit to its own FEL model without being known to others, broadcast the committed FEL model, and then reveal the committed FEL model at a later stage. Therefore, a BCFL node cannot plagiarize an FEL model and commit to it before revealing the FEL model. 
The detailed process is summarized in Algorithm~\ref{al_2} and illustrated in Fig. \ref{fig:flowchart} for better understanding.

% First, each BCFL node $e_i$ will apply the HCDS to its possessed FEL model set $W^{i}(k)= \{w_1^{{i}}(k),w_2^{{i}}(k),...,w_{n_{i}}^{{i}}(k)\}$, where $n_i$ is the total number of FEL models possessed by BCFL node $e_i$ (Line 1 in Algorithm~\ref{al_2}). The details and workflow of the HCDS function are shown in Algorithm~\ref{al_3} and Fig. \ref{fig:HCDS}. 
To begin with, each BCFL node $e_i$ initializes a fixed-length random nonce $r^{i}(k)$ at round $k$ (Line 1), 
% Meanwhile, the HCDS will also generate the digest set $D^{i}(k)=\{d_1^{{i}}(k),d_2^{{i}}(k),...,d_{n_{i}}^{{i}}(k)\}$ and the tag set $TAG^{i}(k)=\{tag_1^{{i}}(k),tag_2^{{i}}(k),...,tag_{n_{i}}^{{i}}(k)\}$. The values of elements in the digest and tag set are initially set to empty (Line 1 in Algorithm~\ref{al_3}). 
which will be concatenated with the FEL model $w^{i}(k)$ into $r^{i}(k) \ ||\  w^{i}(k)$. And then this concatenation will become the input of the Hash function $H$ to derive the digest $d^{i}(k)$ (Line 2). The digest $d^{i}(k)$, together with the private key $SK_{i}$, will be further passed to the signing algorithm $DSign$ to generate the tag $tag^{i}(k)$ (Line 3). %iven the digest $d^{i}(k)$ produced and the private key $SK_{i}$ of the BCFL node $e_i$ (Line 3). 
%We refer to these steps as \textit{Preparation Stage} (Line 1-3 in Algorithm~\ref{al_2}). 
Next, each BCFL node $e_i$ will broadcast its digest $d^{i}(k)$ and tag $tag^{i}(k)$ (Line 4). After all BCFL nodes finish broadcasting, each BCFL node $e_i$ should receive digests and tags from other BCFL nodes, denoted by ${D(k,i)}=\{d^{l}(k) \mid l \neq i \}$ and ${TAG(k,i)}=\{tag^{l}(k) \mid l \neq i \}$, respectively (Line 5). 
%The details and workflow of the Commit function are shown in Algorithm~\ref{al_4} and Fig. \ref{fig:Commit}. 
%We use ${D(k,i)}=\{d^{l}(k) \mid l \neq i \}$ to denote the received digests set at BCFL node $e_i$, where $d^{l}(k)$ represents the digest possessed by the sender BCFL node $e_l$. Similarly, we utilise ${TAG(k,i)}=\{tag^{l}(k) \mid l \neq i \}$ to represent the received tags set at BCFL node $e_i$. 
Then $e_i$ will run the verifying algorithm $DVerify$ to check whether each tag $tag^{l}(k)$ and the associated digest $d^{l}(k)$ are valid using the public key $PK_{l}$ of the sender node $e_l$ (Lines 6-10). We refer to the above steps as \textit{Commit Stage}. % (Lines 1-10). 

For clarity, the workflow of this stage is presented in Fig. \ref{fig:Commit}, where the upper half part corresponds to applying Hash function $H$ and $DSign$ successively; % to the concatenation of its tag $r^i(k)$ and FEL model $w^i(k)$. 
while the bottom half refers to $e_i$ utilizing $DVerify$ to validate each received digest $d^l(k)$ and tag $tag^l(k)$. % Initialization, broadcast, and reception of tags and digests are omitted in the flowchart.

Afterwards, each BCFL node $e_i$ will broadcast its random nonce $r^{i}(k)$, FEL model $w^{i}(k)$ and tag $tag^{i}(k)$ (Line 11), followed by receiving random nonces, FEL models, and tags from other BCFL nodes (Line 12). 
%The details and workflow of the Reveal function are shown in Algorithm~\ref{al_5} and Fig. \ref{fig:Reveal}. 
The random nonces set received by $e_i$ is denoted as $R(k,i)=\{r^{l}(k) \mid l \neq i \}$ and the set of FEL models is $W(k,i)=\{w^{l}(k) \mid l \neq i \}$. Finally, $e_i$ applies the Hash function to the concatenation of received random nonce $r^{l}(k)$ and FEL model $w^{l}(k)$, and the output will be compared with the previously received digest $d^{l}(k)$. If two values match, the $DVerify$ algorithm will be utilized again to check the validity of tag $tag^{l}(k)$ and the output of the Hash function given the public key $PK_{l}$ of the sender BCFL node $e_l$ (Lines 13-19). We refer to these steps as \textit{Reveal Stage}, % (Lines 11-19), 
and the corresponding workflow is presented in Fig. \ref{fig:Reveal} with broadcast and reception steps omitted.

Once the HCDS scheme is accomplished, each BCFL node $e_i$ completes the broadcast and verifies all FEL models. It is worth noting that in addition to the common Digital Signature algorithm with the Hash function, the only extra overhead in our scheme is generating a fix-length random nonce for each FEL model, and thus it is a relatively lightweight scheme.

\begin{figure}[htp]
\centering
% \subfigure[HCDS]{
% \label{fig:HCDS}
% \includegraphics[width=0.45\textwidth]{HCDS_a.pdf}}
\subfigure[Commit Stage]{
\label{fig:Commit}
\includegraphics[width=0.45\textwidth]{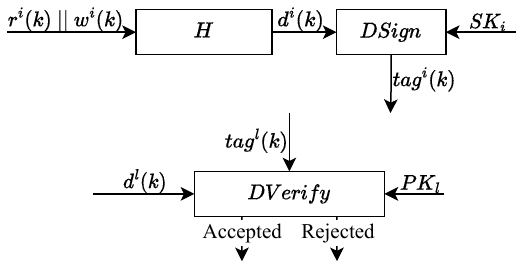}}
\subfigure[Reveal Stage]{
\label{fig:Reveal}
\includegraphics[width=0.45\textwidth]{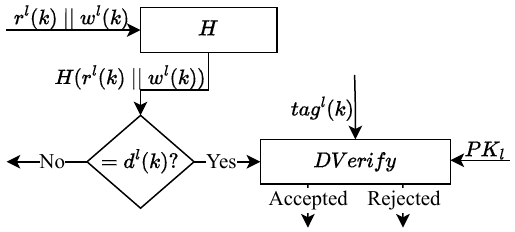}}
\caption{Workflow of Commit Stage and Reveal Stage in HCDS.}
\label{fig:flowchart}
\end{figure}

\subsection{Model Evaluation}
%In our PoFEL consensus, 
Next, we design a lightweight algorithm that reuses the computational resource consumed in finishing FL tasks, namely Model Evaluation (ME), to determine the leader node in each BCFL round based on the similarities between each FEL model and the updated global model.
% As mentioned in the Task publication procedure in section~\ref{sec:sysmod}, the time information of each learning task being published is recorded by each BCFL node $e_i$. Based on the publish time of FL tasks, each BCFL node $e_i$ maintains a time of task set $T^{time} (k) = \{t_1^{time}(k),t_2^{time}(k),...,t_{n_t}^{time}(k)\}$ in ascending time order. The learning task $ts$ used by the ME at each iteration is selected based on the time of task set $T^{time}(k)$ successively and iteratively. 

The details of the ME function are reported in Algorithm~\ref{al_6}. Once the broadcast and verification of all FEL models are completed, each BCFL node $e_i$ will first perform the aggregation of all FEL models in $W(k)$ according to
\begin{align}
    gw(k)=\frac{\sum_{m=1}^{N} \left | DS_m \right | w^m(k)}{\left | DS \right | },
    \label{aggtask}
\end{align}
where $DS=\bigcup_{m=1}^{N} {DS}_m$ (Line 1), with  ${DS}_m$ representing the combination of all clients' local datasets associated with the BCFL node $e_m$,  $\left | DS_m \right |$ is the data size of ${DS}_m$,  %The edge server $e_m$ possesses the FEL model $w^m(k)$. 
and $\left | DS \right |$ refers to the total data size of $DS$.

With the help of ME, each node $e_i$ can calculate cosine similarity\footnote{Other similarity measures can also be employed in our scheme, such as Euclidean distance and $l_2$ norm.} between each FEL model $w^m(k)$ and the updated global model $gw(k)$ to measure their difference %of their directions
\begin{align}
    s_m=\frac{\left \langle w^m(k), gw(k)\right \rangle }{\left \| w^m(k) \right \| \left \| gw(k) \right \|},
    \label{cossim}
\end{align}
where $\left \langle \cdot, \cdot \right \rangle$ means dot product, and $\left \| \cdot \right \|$ refers to $l_2$ norm.

Eventually, of all computed cosine similarities, each BCFL node $e_i$ will cast a vote on the BCFL node with the highest cosine similarity value at round $k$, indicating that its FEL model is the closest to $gw(k)$. We denote the index of this BCFL node as $e^i_{best}(k)$ (Lines 2-5). %In other words, the FEL model from such a BCFL node has the highest cosine similarity value. 

Next, we define  $ p_j^i(k) \in \left [ 0,1 \right ] $ to represent node $e_i$'s prediction about the fraction of voters voting for node $e_j$ as the leader. It is clear that $\sum_{j=1}^{N} p_j^i(k)=1 $. 
Then every $e_i$ can derive the prediction set $P^{i}(k)=\{p_1^{i}(k), p_2^{i}(k), \cdots, p_{N}^{i}(k)\}$ and submit it to the smart contract for tallying votes (Lines 6-13). %Given index $e_{best}^{i}(k)$, 
Now, every $e_i$ will assign a value to $p_j^i(k)$: if $j = e_{best}^{i}(k)$, $p_j^i(k) =G_{max}$ that is close but not equal to 1 (Line 8);  otherwise, $p_j^i(k)=G_{min}$ that is close to 0  (Line 10). % but not equal to 0 (Line 10). 
% \textbf{Since the prediction $p_j^{i}(k)$ at the log component in condition (\ref{predictscore}) will cause errors if $p_j^{i}(k)$ equals zero, here each BCFL node $e_i$ will predict the fraction of votes that the edge server with index $e_{best}^{i}(k)$ will receive to be 0.99. Also due to the condition (\ref{predict}) required by BTSV, the fraction of votes for edge server with index other than $e_{best}^{i}(k)$ is divided equally. In other words, all other predictions will be set to $\frac{1-0.99}{N-1}$ (Lines 6-13).}
\begin{algorithm}
\caption{ME}\label{al_6}

\begin{algorithmic}[1]
\REQUIRE $W(k)$, $DS$
%\STATE \textbf{Initialize}: $P^{i}(k)$$\enspace=\enspace$$\{p_1^{i}(k)=0, p_2^{i}(k)=0, \cdots,p_{N}^{i}(k)=0\}$
\ENSURE $e_{best}^{i}(k)$, $P^{i}(k)$, $gw(k)$
\STATE $gw(k) \gets$ update via (\ref{aggtask})
%\STATE $s_{max} = -\infty$
\FOR{$m = 1$ to $N$}
    \STATE $s_m \gets$ update via (\ref{cossim})
    % \IF {$s_m > s_{max}$}
    % \STATE $e_{best}^{i}(k) \gets m$
    % \STATE $s_{max} \gets s_m$
    % \ENDIF
\ENDFOR
\STATE $e_{best}^{i}(k)\gets$ index of the BCFL node with highest $s_m$ voted by BCFL node $e_i$ at round $k$
%\STATE $best \gets e_{best}^{i}(k)$
%\STATE $p_{best}^{i}(k) \gets 0.99$
\FOR{$j = 1$ to $N$}
\IF{$j = e_{best}^{i}(k)$}
\STATE $p_{j}^{i}(k) \gets G_{max}$
\ELSE
\STATE $p_{j}^{i}(k) \gets G_{min}$
\ENDIF
\ENDFOR
\RETURN $e_{best}^{i}(k)$, $P^{i}(k)$, $gw(k)$

\end{algorithmic}
\end{algorithm}
\subsection{Bayesian Truth Serum-based Voting (BTSV) for Leader} 
\label{sec:BTS}
Now, we need to tally votes from all BCFL nodes and elect the leader node $e^*(k)$ with the highest votes. As analyzed in Section \ref{subsec:bribery}, however, a plain voting scheme can suffer from the threat of bribery voting, where a malicious node may corrupt or bribe some other nodes to be voted, interrupting the tallying integrity and leading to biased results. Therefore, BTSV is developed to count votes in the form of a smart contract. BTS \cite{prelec2004bayesian} is a scoring approach that assigns higher scores to voters picking the common options. The underlying rationale is if a voter indeed has a preference for an option, s/he tends to believe others have similar preferences. % (assigns a higher score to the edge server that casts a vote frankly). 
As for a voter with a relatively low cumulative score, predefined penalties will be incurred, such as reward reduction %acquired through the incentive mechanism 
or lower voting weight %being assigned with the vote 
during tallying. 

In our framework, we design BTSV that associates the cumulative score with the voting weight of each BCFL node $e_i$ when tallying and thus motivates BCFL nodes to vote honestly. Initially, the smart contract processes the vote $e_{best}^{i}(k)$ and vote predictions $P^{i}(k)$ from BCFL node $e_i$ to compute its BTS score. The BTS score reflects the behavior of the BCFL node, rewarding honest nodes with higher scores and penalizing malicious nodes with lower scores (further analysis is provided in Section \ref{sec:briberyvoting}). Subsequently, the system determines the voting power of each BCFL node $e_i$ by considering its cumulative historical BTS score. A sigmoid function is employed to transform this cumulative BTS score into a weight, ensuring the voting weight is appropriately scaled and constrained within a specific range. % (greater than 0 and lower than a maximum). 
Specifically, when the cumulative BTS score is lower, the weight approaches 0, and as the cumulative BTS score increases, the weight also increases to approaching the maximum value. This mechanism ensures that nodes with higher cumulative BTS scores have stronger voting power, while those with lower scores have a limited influence on the voting process. Finally, the votes are tallied, and the leader node with the highest adjusted votes is selected. The detailed procedures are summarized in Algorithm~\ref{al_7}.
%The BCFL node $e_i$ with a relatively low cumulative score could get predefined penalties such as reductions in rewards acquired through the incentive mechanism or lower weight being assigned with the vote during tallying. The BTSV provides a scoring system that 
 
\begin{algorithm}
\caption{BTSV}\label{al_7}

\begin{algorithmic}[1]
\REQUIRE $E_{best}(k)$, $P(k)$
\ENSURE $e^{*}(k)$

\FOR{$j = 1$ to $N$}
\FOR{$i = 1$ to $N$}
    \IF {$j$ = $e^i_{best}(k)$}
    \STATE $A_j^i(k)\gets$ 1
    \ELSE
    \STATE $A_j^i(k)\gets$ 0
    \ENDIF
\ENDFOR
\STATE $\overline{x}_j \gets$  (\ref{xa})
\STATE $\overline{y}_j \gets$ (\ref{ya})

\ENDFOR
\FOR{$i = 1$ to $N$}
\STATE $information \ score \gets$ (\ref{informationscore})
\STATE $prediction \ score \gets$ (\ref{predictscore})
\STATE $score^i(k) \gets$ \eqref{scorei} %$information \ score + prediction \ score$
\STATE $CHS^i(k) \gets$ \eqref{chs}
\STATE $WV^i(k) \gets$ \eqref{wv}
\ENDFOR
\FOR{$j = 1$ to $N$}
\STATE $advotes_j \gets$ \eqref{advote}
\ENDFOR
\STATE $e^{*}(k) \gets $index of the BCFL node with the highest $advotes_j$ at round $k$
% \STATE $\overline{x}_{max} \gets -\infty$
% \FOR{$j = 1$ to $N$}
% \IF {$\overline{x}_j > \overline{x}_{max}$}
%     \STATE $e^{*}(k) \gets $ edge server with index $j$
%     \STATE $\overline{x}_{max} \gets \overline{x}_j$
% \ENDIF
% \ENDFOR
\RETURN $e^{*}(k)$
\end{algorithmic}
\end{algorithm}

 %In our proposed BTSV, the vote $e_{best}^{i}(k)$ from each BCFL node $e_i$ endorse a leader node option. 
We first define 
\begin{align}
    A_j^i(k) =
    \begin{cases}
      1 & \text{if $e_i$ votes for $e_j$ at round $k$},\\
      0 & \text{otherwise} .
    \end{cases} \notag 
    %\label{B}
\end{align}
Then we calculate
\begin{align}
    \overline{x}_j=\frac{1}{N}\sum_{i=1}^{N}A_j^i(k),
    \label{xa}
\end{align}
and
\begin{align}
    \overline{y}_j=\exp(\frac{1}{N}\sum_{i=1}^{N}\log{p_j^i(k)}),
    \label{ya}
\end{align}
where $\overline{x}_j$ represents the percentage of votes that BCFL node $e_j$ actually receives and $\overline{y}_j$ means the average predicted percentage of votes that BCFL node $e_j$ receives (Lines 1-11).
The score of the BCFL node $e_i$ for a vote for $e_j$ is defined by $\log{\frac{\overline{x}_j}{\overline{y}_j}}$.
Thus, the information score for the BCFL node $e_i$ is calculated as follows
\begin{align}
    information \ score = \sum_{j=1}^{N} A_j^i(k)\log{\frac{\overline{x}_j}{\overline{y}_j}}, 
    \label{informationscore}
\end{align}
and the prediction score is denoted by
\begin{align}
    prediction \ score = \alpha \sum_{j=1}^{N}\overline{x}_j\log{\frac{p_j^i(k)}{\overline{x}_j}}.
    \label{predictscore}
\end{align}
The total score of $e_i$ at round $k$, denoted as $score^i(k)$, is calculated as follows
\begin{align}
score^i(k)& = information \ score + prediction \ score, \notag \\ 
&= \sum_{j=1}^{N} A_j^i(k)\log{\frac{\overline{x}_j}{\overline{y}_j}} + \alpha \sum_{j=1}^{N}\overline{x}_j\log{\frac{p_j^i(k)}{\overline{x}_j}}. 
    \label{scorei}
\end{align}
We let $\alpha = 1$ so that this vote tallying process can be regarded as a zero-sum game. After deriving $score^i(k)$ for $e_i$, the cumulative historical score (CHS) of $e_i$ at round $k$, denoted as $CHS^i(k)$, is computed by
\begin{align}
CHS^i(k) = \sum_{max(0, k-c)}^{k}score^{i}(k),
    \label{chs}
\end{align}
where $c$ is configurable and refers to the number of past rounds considered in determining CHS. Then a sigmoid function is applied to derive the weight of the vote (WV) that $e_i$ casts at round $k$, denoted as $WV^i(k)$, as follows,
\begin{align}
WV^i(k) = \frac{\beta}{1+e^{-\theta CHS^i(k)-\epsilon}},
    \label{wv}
\end{align}
where $\beta$, $\theta$, and $\epsilon$ are positive coefficients (Lines 12-18). Specifically, $\beta$ determines the upper limit of WV, $\theta$ decides the gradient of WV against CHS (i.e., how quickly WV decreases when CHS drops), and $\epsilon$ assigns the value of WV as 1 when CHS is 0 since it implies that the node has never been rewarded or penalized before. % Therefore, given the above meanings of the coefficients, $\beta$ and $\theta$ must be positive. And Given determined $\beta$ and $\theta$, the value of $\epsilon$ needs to satisfy the condition that the value of voting weight equals 1 (or is very close to 1) when CHS is 0. 
A sigmoid function is chosen as it limits the range of WV with an "S" shape, where too-small or too-large CHS results in a smoother change of WV. This aligns with our objective of strongly rewarding or penalizing nodes when they start showing good or bad behaviors instead of reaching extreme cases. 
%moves toward negative infinity, the WV diminishes at an accelerating rate followed by a gradual slowdown towards 0, yet never reaching it. This choice aligns with  our objective of progressively imposing stronger penalties on a node as its CHS decreases, and subsequently tapering down the penalties as its WV diminishes to a lower value (when adequate penalty has been applied). 
Based on $WV^i(k)$ assigned to each BCFL node voter $e_i$, adjusted tallied votes $advotes_j$ that BCFL node $e_j$ receives are computed by (Lines 19-21)
\begin{align}
advotes_j = \sum_{i=1}^{N}WV^i(k)A_j^i(k).
    \label{advote}
\end{align}
Eventually, the leader node with the highest adjusted tallied votes $advotes_j$ at round $k$ will be selected. We denoted the index of this leader node as $e^{*}(k)$ (Line 22).
%The aforementioned procedures are summarized in Algorithm~\ref{al_7}.
% \subsection{Model Aggregations by Leader}
% After the leader node $e^*(k)$ is selected, the leader node will apply the function AG to aggregate FEL models working on each remaining learning task $t_b$ into updated global models by
% \begin{align}
%     gw^{t_b}(k)=\frac{\sum_{m=1}^{n_{t_b}} \left | DS_m^{t_b} \right | w_m^{t_b}(k)}{\left | DS^{t_b} \right | }, \notag
%     %\label{aggtask1}
% \end{align}
% where index $b \in n_t$, $t_b \neq ts$, and $n_{t_b}$ refer to the total number of FEL models (nodes) working on the learning task $t_b$.

% Eventually, the leader node $e^*(k)$ will gather each updated global model $gw^{t_b}(k)$ and the global model $gw^{ts}(k)$ into updated global models set $GW(k)$. Then, the leader node will append the information, including leader node $e^*(k)$, all FEL models in set $W(k)$, and all updated global models in set $GW(k)$ into the New Block. The New Block will be further broadcast to other BCFL nodes.

\section{Incentive Mechanism}
\label{sec:incent}
To motivate end devices to contribute to FEL tasks continuously and encourage BCFL nodes to maintain the blockchain, block generation rewards to BCFL nodes and FEL rewards to the FEL clusters are required. The leader node selected by the consensus algorithm at each BCFL round will acquire a fixed number of rewards. These rewards come from the task publisher’s budget and could be determined when the learning task is initially published. Thus, in this section, we mainly focus on designing an incentive scheme to distribute FEL rewards through BCFL nodes from the FL task publisher to the clients contributing to the learning task.

% \subsection{Block Generation Rewards}
% The leader node selected by the consensus algorithm at each new block generation round will acquire a fixed number of rewards. These rewards come from the associated learning task publisher’s budget and could be determined when the task publisher publishes the learning task. 

% Since every BCFL node is involved in the consensus algorithm at each new block generation round, BCFL nodes not participating in the currently selected learning task $ts$ will never have a chance to become the leader. Thus, they can never receive the block generation rewards at the current new block generation round.
% However, such a phenomenon does not necessarily lead to a lack of motivation for these BCFL nodes to perform the consensus algorithm. As long as FEL networks (edge servers and their connected clients) work on any learning task, they receive the FEL rewards. In the meantime, BCFL nodes (edge servers) will get a chance to compete for the leader node once the execution of the consensus algorithm iterates to the learning task they are involved in. 

% Also, capable BCFL nodes could be working on multiple learning tasks simultaneously. The more learning tasks a BCFL node is involved in, the more chance it will have to become a leader node. Thus giving block generation rewards only to the leader node will not cause a lack of motivation for BCFL nodes that perform the consensus algorithm but are not participating in the current learning task selected by the consensus algorithm.

%\subsection{FEL Learning Rewards}
The task publisher will provide rewards to motivate FEL clusters to work on the requested learning task. Rewards for an FEL cluster are first determined between the task publisher and the BCFL node, who will distribute rewards to end devices according to a pre-defined rule. An example distribution rule could be based on the CPU cycle frequency spent by each end device. %Such rules are decided by the edge server and its associated end devices, thus each FEL cluster could have different end device rewards distribution rules. 
In our framework, the major challenge of incentive mechanism exists at the blockchain level between the task publisher and BCFL nodes.

To that aim, we model a two-stage Stackelberg game between the task publisher and each participating BCFL node before running FEL, which can be described as below:
%where their equilibrium strategies can be derived. 
%Specifically, the task publisher will determine the total rewards assigned to all edge servers working on its published learning task, and then each edge server will decide the total CPU cycle frequency (cycles per second) of its connected clients according to the received rewards.  Therefore, we could model the process of the two-stage Stackelberg game as follows:
\begin{itemize}
    \item Stage 1: The task publisher will set the total rewards $\delta$ for all engaging BCFL nodes by maximizing its own utility. This utility function is based on its expectation, expenditure (rewards $\delta$), and total CPU cycle frequencies consumed by all engaging BCFL nodes.
    \item Stage 2: After calculating received individual rewards, each BCFL node $e_i$ will the total CPU cycle frequency $f_i$ spent by its connected devices by optimizing its individual utility.
\end{itemize}

Backward induction is adopted to solve the two-stage Stackelberg game. In other words, we will start with optimizing the strategy of Stage 2 first and then analyze the best strategy of Stage 1. %The utility models of task publisher and edge servers are defined below:
\subsection{Utility Function}
The task publisher's expectation of the total rewards $\delta$ paid to all engaging BCFL nodes is closely related to the total CPU cycle frequencies consumed by all FEL clusters, denoted by $F= \sum_{i=1}^{N}f_i$. In other words, given the total CPU cycle frequencies $F$, the task publisher can derive the optimal total rewards, denoted by $\delta^*$, that it believes to be the best total price paid to all participating BCFL nodes. Any rewards different from the optimal will lead to a decrease in the task publisher's utility. This relationship could be captured by a parabola that opens downwards. Thus, the utility of the task publisher is defined as
\begin{align}
    U_{tp} (\delta)= B - (\lambda \frac{\delta}{F}-\varphi)^2,
    \label{Um}
    \end{align}
where $B$, $\lambda$, and $\varphi$ are coefficients reflecting the task publisher's expectation of total rewards $\delta$ given $F$. Mathematically speaking, they determine the shape, location, and opening direction of the parabola in the plane. %Total CPU cycle frequencies F is defined as $F = \sum_{i=1}^{N}f_i$.
    % \begin{align}
    % , \notag
    % %\label{F}
    % \end{align}
    %where $n_{tp}$ is the number of edge servers working on the task publisher's task. 
    The utility model defined in (\ref{Um}) needs to satisfy $U_{tp} (\delta) > 0$ and $\frac{\varphi}{\lambda} > 0$.
   %  \begin{align}
   %  when \ \delta = 0,\ U_{tp} (\delta) > 0,
   %  \label{condition1}
   %  \end{align}
   %  \begin{align}
   % \arg \max_{\delta} U_{tp} (\delta) > 0.
   %  \label{condition2}
   %  \end{align}
   % \begin{align}
   % U_{tp} (\delta) > 0.
   % \label{condition1}
   % \end{align}
    By solving the first constraint, we have $-\sqrt{B} < \lambda \frac{\delta}{F}-\varphi < \sqrt{B}$.
    These inequations refer to the model owner's utility being positive when optimal rewards are given. 
    
%We let $f_i$ be the total CPU cycle frequency of the clients connected to the edge server $e_i$. 
We assume that the total CPU cycles used for training data samples in the learning task are $\mu_i$. Then the energy cost of computing the task could be defined as $\gamma_i\mu_if_i^2$ based on the popular energy consumption model \cite{burd1996processor}. Parameter $\gamma_i$ is related to the CPU's architecture. 
 
\if()
\begin{align}
    C_i = \frac{\sum_{k=1}^{K}s_i(k)}{\sum_{k=1}^{K}\sum_{i=1}^{N}s_i(k)}, \notag
    %\label{contribution}
    \end{align}
\fi
Therefore, we can derive the utility of BCFL node $e_i$ as
    \begin{align}
    U_{i} (f_i)= \delta \frac{f_i}{f_i+\sum f_{-i}}  - \gamma_i\mu_if_i^2,
    \label{Ui}
    \end{align}
where $\sum f_{-i}$ refers to the summation of CPU cycle frequencies of all BCFL nodes except for the CPU cycle frequency $f_i$ of the BCFL node $e_i$ and can be considered as a constant.

    % \begin{align}
    % when \ f_i = 0,\ U_{i} (f_i) > 0
    % \label{condition3}
    % \end{align}
   %  \begin{align}
   % \arg \max_{f_i} U_{i} (f_i) > 0.
   %  \label{condition4}
   %  \end{align}
     %Solving the condition 
     %(\ref{condition3}) and 
     %(\ref{condition4}), we have 
     %$\delta * C_i> \gamma_i*\mu_i*{\psi_i}^2$ 
     %$\frac{\psi_i}{\beta_i}>0$. Similarly, the condition 
     %(\ref{condition3}) and 
     %(\ref{condition4}) means that edge server $e_i$ 's utility will still be positive when the edge server spends 
     %zero and 
     %optimal total CPU frequency $f_i$.
\subsection{Equilibrium Analysis}
In this subsection, we will acquire the Nash equilibrium from deriving the optimal strategies for both the task publisher and BCFL nodes. 
\begin{definition}
(Stackelberg Equilibrium Strategy).
The Stackelberg equilibrium strategies ($\delta^*$, $f_i^*$) formulate the best response of each player in the Stackelberg game, where $\delta^* = \arg \max_{\delta > 0} U_{tp} (\delta)$ and $f_i^* = \arg \max_{f_i > 0} U_{i} (f_i)$.
% \begin{align}
%    & \delta^* = \arg \max_{\delta > 0} U_{tp} (\delta), \notag \\
%    & f_i^* = \arg \max_{f_i > 0} U_{i} (f_i). \notag
%    %\label{Eqstra}
% \end{align}
\end{definition}
The second stage of the game is a non-cooperative game. For any total rewards $\delta$ given by the task publisher and other BCFL node’s strategies, BCFL node $e_i$ will decide an optimal strategy $f_i^*$ to maximize its individual utility. 

\begin{definition}
(Nash Equilibrium). A set of strategies $f^* = (f_1^*, \cdots, {f_{N}}^* )$ is a Nash equilibrium of the second stage if for each BCFL node $e_i$ and any $f_i$, there exists $U_{i} (f_i^*)\geq U_{i} (f_i)$.
% \begin{align}
%   U_{i} (f_i^*)\geq U_{i} (f_i). \notag
%    %\label{Nasheq}
% \end{align}
\end{definition}

\begin{theorem}
\label{theorem1}
The optimal strategy $f_i^*$ of BCFL node $e_i$ is derived by letting the first order derivative equal 0. %The derived equation is 
% \begin{align}
%   2 \mu_i\gamma_i\beta_i^2 f_i^*+ 4\mu_i\gamma_i\beta_i {f_i^*}^2 + 2\mu_i\gamma_i{f_i^*}^3-\delta\beta_i=0. \notag
%     %\label{solu1}
%     \end{align}
\end{theorem}
\begin{proof} To study the Nash equilibrium of the second stage of the game, the first-order derivative of $U_{i} (f_i)$ with respect to $f_i$ is derived as
\begin{align}
    \pdv{U_{i} (f_i)}{f_i}= \frac{\delta}{\sum f_{-i}+f_i}-2f_i\mu_i\gamma_i-\frac{\delta f_i}{(\sum f_{-i}+f_i)^2}. \notag
    %\label{FirstdevUi}
    \end{align}
The second-order derivative of $U_{i} (f_i)$ with respect to $f_i$ is derived as
\begin{align}
    \pdv[2]{U_{i} (f_i)}{f_i}= \frac{2\delta }{(\sum f_{-i}+f_i)^2} (\frac{f_i}{\sum f_{-i}+f_i}-1) -2\mu_i\gamma_i, \notag
    %\label{SeconddevUi}
    \end{align}
which is negative since $\frac{f_i}{\sum f_{-i}+f_i}$ is always less than 1.
Therefore, the utility of BCFL node $e_i$ defined in (\ref{Ui}) is a strictly concave function for $f_i > 0$. And there exists  $f_i^*$ that maximizes $U_i$, which can be derived via solving $\pdv{U_{i} (f_i)}{f_i}=0$.

% By letting $\pdv{U_{i} (f_i)}{f_i}=0$, we can have
% \begin{align}
%   f_i^*= \frac{\psi_i}{\beta_i}. \notag
%     %\label{solu1}
%     \end{align}
\end{proof}

Then the optimal value of total CPU cycle frequencies $F$ can be calculated by 
 \begin{align}
    F^* = \sum_{i=1}^{N} f_i^*.
    \label{F*}
    \end{align}
    
Given the above analysis, the task publisher knows that there exists a unique Nash equilibrium among BCFL nodes given any total rewards $\delta$. Therefore, the task publisher will maximize its utility by choosing the optimal total rewards $\delta^*$.

\begin{theorem}
\label{theorem2}
The optimal strategy of the task publisher is given 
\begin{align}
    \delta^* =  \frac{{F^*}\varphi}{\lambda}.
    \label{Delta*}
    \end{align}
\end{theorem}

\begin{proof} The first-order derivative of $U_{tp} (\delta)$ with respect to $\delta$ is derived as
\begin{align}
   \pdv{U_{tp} (\delta)}{\delta}= \frac{2\lambda(F^{*}\varphi-\delta\lambda)}{{F^{*}}^2}. \notag
    %\label{FirstdevUtp}
    \end{align}
Then the second-order derivative of $U_{tp} (\delta)$ with respect to $\delta$ is derived as
\begin{align}
   \pdv[2]{U_{tp} (\delta)}{\delta}= -\frac{2\lambda^2}{F^2}, \notag
    %\label{SeconddevUtp}
    \end{align}
which is clearly negative. 
Therefore, the utility of task publisher defined in (\ref{Um}) is a strictly concave function for $\delta > 0$. That is, there exists a $\delta^*$ maximizing $U_{tp}$.    

By letting $\pdv{U_{tp} (\delta)}{\delta}=0$, we have
    \begin{align}
  \delta^* =  \frac{{F^*}\varphi}{\lambda}. \notag
    %\label{solu1}
    \end{align}
\end{proof}

% The proofs of \textbf{Theorems \ref{theorem1} and \ref{theorem2}} are straightforward and only need to calculate the first- and second-order derivatives to determine the convexity of the utility functions to get the optimal solutions. Due to the space limitation, the detailed proofs are omitted.
% \section{Simulation Setup}
% \label{sec:simsetup}

\section{Security Analysis}
\label{sec:security}
\subsection{Model Plagiarism} 
%The implemented HCDS scheme is used to prevent potential plagiarism. 
In BCFL, the consensus algorithm usually requires all BCFL nodes to exchange their models in the first step for model aggregation, leader election, and verification. However, in this sort of distributed network, it is difficult to achieve model exchanges at all nodes at the same time. 
%each node's broadcasting and receiving processes could take a different amount of time. 
In other words, some nodes could receive models much earlier than others. This provides a chance for malicious nodes that received models earlier to plagiarize models so that they could obtain rewards without consuming their own resources.
%energy used to perform their FEL learning while copying another BCFL node' model and claiming it to be their own. 
To solve this problem, we design %a scheme to let BCFL nodes find that plagiarism is extremely difficult or not cost-effective when compared with behaving normally to perform individual FEL learning. So we used the idea of 
HCDS scheme to make the model plagiarism impossible, or not worthy compared to behaving legally, where the employed Hash-based Commitment computationally hides % in terms of recovering a portion of 
any model $w$. Furthermore, even if the adversary can find $r\ ||\ w$ from $H(r\ ||\ w)$ given enormous computation power, the exact value of $w$ still cannot be derived. This means it is practically impossible to plagiarize any models from others after the Commit Stage in Algorithm~\ref{al_2}. %Hash-based Commitment Scheme is computationally hiding but not binding, but in our case, the goal is to avoid plagiarism. So 

Regarding the binding constraint, even if a malicious node tries to broadcast $r^{\prime}\ ||\ w^{\prime}$ at Reveal Stage instead of the original committed $r\ ||\ w$ at Commit Stage, the chance that revealed $r^{\prime}\ ||\ w^{\prime}$ equals to committed $ r\ ||\ w$ and $w^\prime$ coincidentally matches a model weight committed earlier by another BCFL node is extremely low. If the malicious node finds it challenging to obtain a valid copy $w^\prime$ from another BCFL node such that $r^{\prime}\ ||\ w^{\prime} $ equals to $ r\ ||\ w$, and it is computationally infeasible to infer $w$ from $H(r\ ||\ w)$, the BCFL node will be demotivated to plagiarize.

\subsection{DDoS} Since the leader is determined by the consensus algorithm based on the model similarity evaluation and trustworthy voting, it is impossible to know in advance about the leader node that is highly possible to change every round. In this case, a constantly effective DDoS attack against the leader node cannot be achieved.

\subsection{Bribery Voting} 
\label{sec:briberyvoting}
In the proposed consensus algorithm, vote tallying is involved in deriving the leader node. In this case, a malicious node could bribe other BCFL nodes to cast dishonest votes to interfere with tallying results. Thus, we introduce the BTS into the weighted voting process to encourage BCFL nodes to behave honestly. Given that BCFL node $e_i$ votes for an option, its information score would be lower if the actual fraction of votes received by this option is lower than the average predicted fraction of votes that this option receives. And the prediction score gives a penalty if the BCFL node $e_i$'s predicted fraction of votes that an option receives is lower than the actual fraction of votes. The more commonly this option is voted, the more penalty is given. 

Both malicious and honest nodes will cast a vote for an option, and assign a predicted fraction of votes this option receives to $G_{max}$ while other options as $G_{min}$. This means that for every option, the actual fraction of votes it receives is always roughly the same as the average predicted fraction of votes it receives. The information scores for both malicious and honest nodes are not distinguishable. But malicious nodes that cast votes for a BCFL node other than the correct one with the highest $s_m$ will be penalized by prediction score as it receives a huge penalty for low prediction on the votes the correct BCFL node should receive (as the correct BCFL node receives a higher fraction of votes). Thus, the overall score of malicious nodes will be lower than that of honest nodes. Then the system will impose lower weights on the votes cast by malicious nodes since their CHS tend to be lower. Lower weights on the votes reduce the voting power of the malicious nodes and weaken their impacts on the vote tallying process.

\section{Experimental Evaluation}
\label{sec:experiment}
In this section, we perform several experiments to verify the effectiveness and efficiency of our proposed consensus algorithm and incentive mechanism. 
 %First, we evaluate the performance of components from the Commitment Scheme and the ME mechanism. Then we analyze the effectiveness of the BTSV and the Stackelberg game-based incentive mechanism. 
 All experiments are conducted using Google Colab with Python 3.6.9 running on single core hyper threaded  Intel(R) Xeon(R) CPU at 2.3GHz, Tesla T4 GPU, and 13GB memory. 
 
\subsection{Experiment Setups and Dataset}
We implement the BHFL system where each FEL cluster involves five local clients. Each BCFL node will perform FEL with its associated clients three times before exchanging FEL models with other BCFL nodes at the BCFL system. The default number of BCFL nodes involved is set to 50. 

1) \textbf{Dataset:} The FL task is conducted on MNIST \cite{lecun1998gradient} dataset using the MLP model. The MNIST dataset is a collection of handwritten digits that is widely used as a benchmark for image recognition algorithms. It consists of 60,000 training images and 10,000 testing images, with each image being a grayscale 28x28 pixel image of a handwritten digit from 0 to 9. All training images are used to train the model and randomly distributed to every end device.

2) \textbf{Model:} We utilize a multilayer perception (MLP) model which consists of a flatten layer, a hidden layer with ReLU activation, a dropout layer with a frequency of 0.2, and an output layer with softmax activation. The flatten layer is used to map the original image data to a one-dimensional array that contains 784 elements (28*28) which can be read by the hidden layer. The hidden layer is made up of 128 neurons by default. The output layer contains 10 neurons that correspond to the 10 labels in the MINST dataset. The stochastic gradient descent (SGD) optimizer is used with a learning rate of 0.001, a decay factor equal to half of the learning rate, and a momentum value of 0.9.

%Every FEL network works on the FL task in this evaluation.
\subsection{Evaluation of HCDS}
%Here we provide a numerical analysis of the computation cost of the HCDS affected by various factors. 
HCDS consists of two stages: Commit Stage and Reveal Stage, and we evaluate them separately. The Hash function and Digital Signature Algorithm used in the experiment are SHA-256 and Elliptic Curve Digital Signature Algorithm (ECDSA) \cite{johnson2001elliptic}, respectively. The results reported are derived by averaging results for ten rounds of repeated experiments. % for statistical confidence.

We first investigate the computation cost in Commit Stage, which involves three computational functions: Hash function $H$, $DSign$, and $DVerify$. We first study the impacts of MLP model complexity and random nonce length on Hash function $H$ and $DSign$. 
%As explained in Algorithm~\ref{al_2}, each BCFL node $e_i$ will apply Hash function H and DSign to its possessed FEL model $w^i(k)$. 
%The number of edge servers in the BCFL system will not affect the computation cost. 
As shown in Lines 2-3 in Algorithm~\ref{al_2}, %the length of the output (digest) from the Hash function is fixed regardless of the MLP model and random nonce. The time cost for the signing algorithm DSign is constant, given that the input digest length is the same.
the only factors that will affect the time cost of $H$ and $DSign$ are inputs $r^i(k)$ and $w^i(k)$, which are specifically the random nonce length and the MLP model complexity. To model changes in the complexity of the MLP model, we change the number of neurons in the hidden layer in the MLP model. 
As shown in Fig. \ref{fig_3_hbcds}, the time cost increases linearly with the length of the random nonce length. Furthermore, the rise in the amount of time cost caused by the change in model complexity is almost constant regardless of the random nonce length given. All these trends are consistent with the inherent attributes of the selected Hash function $H$ and Digital Signature Algorithm.

\begin{figure}[ht]
\centering
\subfigure[Model Complexity and Random Nonce Length vs. $H$ and $DSign$ cost.]{
\label{fig_3_hbcds}
\includegraphics[width=0.23\textwidth]{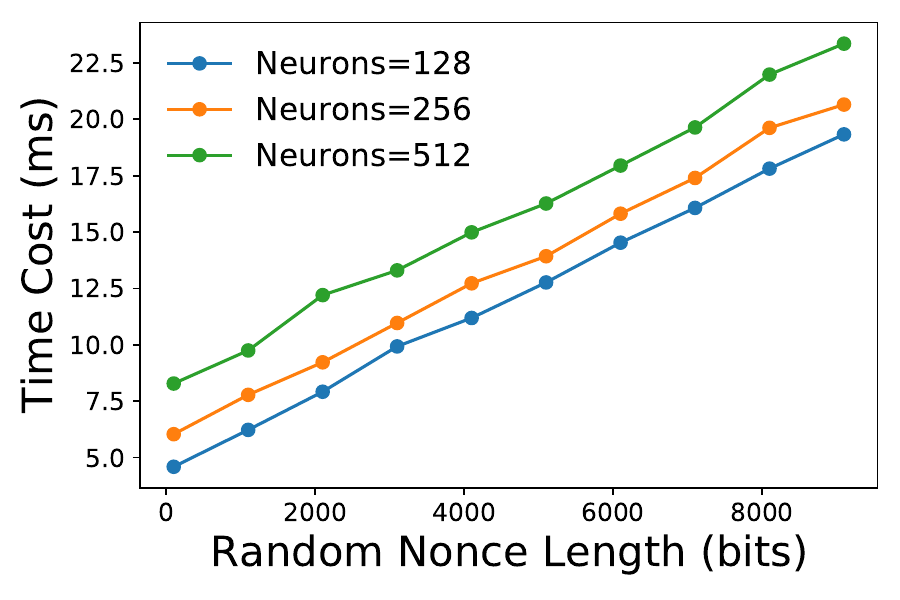}}
\subfigure[Network Size vs. $DVerify$ cost.]{
\label{fig_3_commit}
\includegraphics[width=0.23\textwidth]{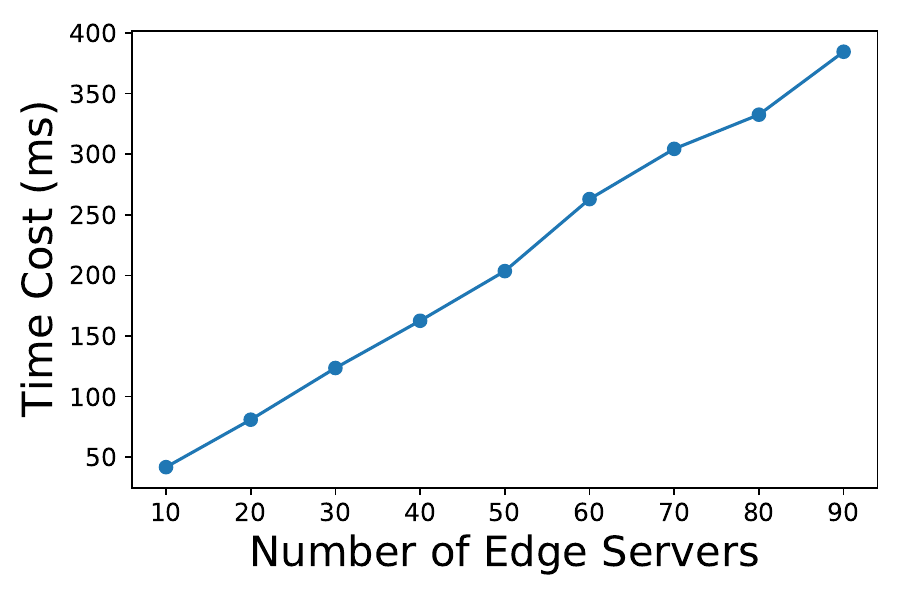}}
\caption{Computation Cost in Commit Stage.}
\label{fig_3}
\end{figure}

Then we examine the impact of BCFL network size $N$ on the computation cost of $DVerify$. The result is illustrated in Fig. \ref{fig_3_commit}. More  BCFL nodes in the network bring more  FEL models for $DVerify$ to execute at BCFL node $e_i$. Similarly, a linear relationship between the time cost and the number of BCFL nodes is observed, and such a relationship matches with the Digital Signature Algorithm used.

Finally, we analyze the impact of BCFL network size, random nonce length, and model complexity on the computation cost in Reveal Stage, which involves the Hash value computations and $DVerify$. Thus all previously mentioned factors will impact the computation cost in Reveal Stage. Fig. \ref{fig_4_1} demonstrates a linear trend between the number of BCFL nodes and time cost. We also notice that the impact of the random nonce is inconspicuous on the relationship between the number of BCFL nodes and time cost. We could infer the reason by comparing Fig. \ref{fig_3_commit} and Fig. \ref{fig_4_1}. % Fig. \ref{fig_3_commit} represents the time taken by $DVerify$, and Fig. \ref{fig_4_1} represents the time taken by $DVerify$ and Hash function $H$. 
The random nonce length will affect the time taken by Hash function $H$ but not $DVerify$. And the inputs of $DVerify$ are tags and digests, which are all fixed given the same Hash function $H$ and $DSign$. This could imply that the amount of time increased in Hash function $H$ due to the change of the random nonce length is minor compared with the total amount of time needed by $DVerify$ and Hash function $H$.

On the other hand, the model complexity gradually impacts the trend between the time cost and the number of BCFL nodes, as shown in Fig. \ref{fig_4_2}. The more significant number of neurons leads to more time cost computing the Hash value for a single FEL model, and more BCFL nodes mean more FEL models to process. Therefore, a steeper trend between the time cost and the number of BCFL nodes is observed when the number of neurons is more considerable (more accumulated additional time cost of a single FEL model).

\begin{figure}[ht]
\centering
\subfigure[Random Nonce Length \& Network Size.]{
\label{fig_4_1}
\includegraphics[width=0.23\textwidth]{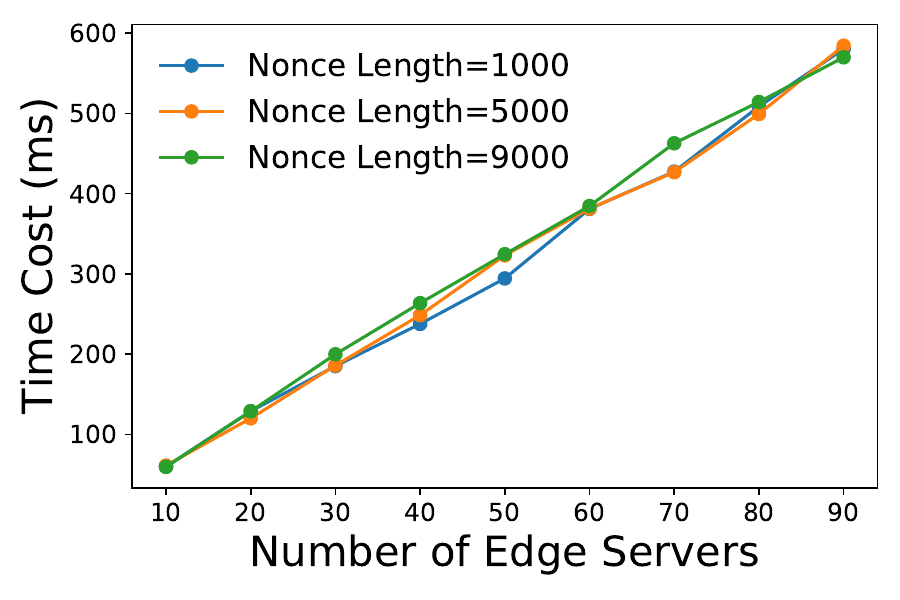}}
\subfigure[Model Complexity \& Network Size.]{
\label{fig_4_2}
\includegraphics[width=0.23\textwidth]{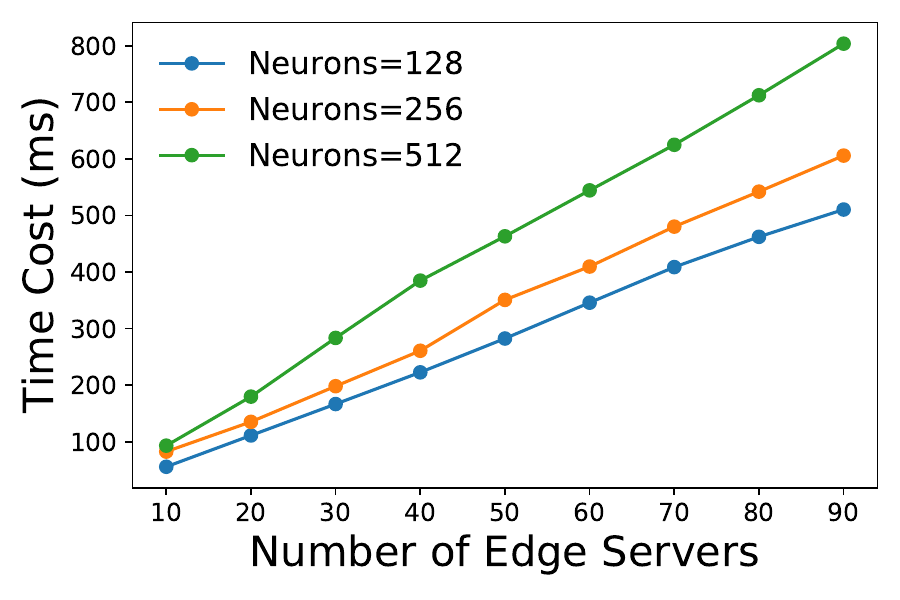}}
\caption{Computation Cost in Reveal Stage.}
\label{fig_4}
\end{figure}
\subsection{Experiments on Model Evaluation}
In this subsection, experiments are designed to analyze the computation cost and the randomness of the ME mechanism. Randomness refers to the difference in the chance of each BCFL node participating in the same FL task being selected as the leader. The randomness metric directly reflects the fairness of ME. The input of the ME scheme is FEL models. Therefore, the number of BCFL nodes and the model complexity will affect its computation time. % of the ME mechanism.

First, we will investigate the impact of the number of BCFL nodes and the model complexity on the time cost of the ME mechanism. Experimental results are plotted in Fig. \ref{fig_5_1}. The number of BCFL nodes and the time cost of ME demonstrates a linear trend. We can see that when the number of neurons is small, the trend between the number of BCFL nodes and the time cost is mild. This indicates the same reason as explained in Fig. \ref{fig_4_2} that more neurons cause more time for ME to compute the similarity of one FEL model. Thus, more BCFL nodes (more FEL models to proceed) lead to increased accumulated time costs.

Then we evaluate the randomness of our proposed ME scheme. % with the distribution of the client's data. 
We implement two distribution types of the clients' data: independent and identically distributed (IID) and non-IID. For non-IID, each client possesses data with roughly six out of ten labels in the MNIST dataset. For IID, data with all labels are available to each client. Fig. \ref{fig_5_2} presents the number of times each BCFL node is selected as the leader before the global model of the FL task is converged. Observations indicate less fairness when the data distribution is non-IID. This could be caused by the dominance of several BCFL nodes having more labels available on their associated clients' local data compared with other BCFL nodes. Given more diverse data, the FEL model learned is more representative. The chance of the BCFL node being selected as the leader increases. The randomness of ME is closely related to the actual distribution of the data possessed by each client.
% Then we evaluate the randomness of our proposed ME scheme with the different number of edge servers and distribution of the client's data possession. We implement two distribution types: IID and non-IID. For non-IID, each client possesses data with roughly six out of ten labels in the MNIST dataset. As illustrated in Fig. \ref{fig_5_2}, the ME scheme could achieve considerable fairness when the number of edge servers is significant. The fairness is poor when the data distribution is non-IID. This could be caused by the dominance of one edge server having the most data labels on its associated clients. Also, we suspect that the reason why fairness is achieved better when the number of edge servers is considerable is that in our experiment, the total number of labels available is ten. The total number of data available is limited. Therefore when the number of edge servers is vast, the data possessed by each client will become much smaller. As a result, the variance of the FEL model weights learned will increase significantly. The performance of the ME scheme is directly related to the variance of the FEL model weights.
\begin{figure}[ht]
\centering
\subfigure[Model Complexity and Network Size.]{
\label{fig_5_1}
\includegraphics[width=0.23\textwidth]{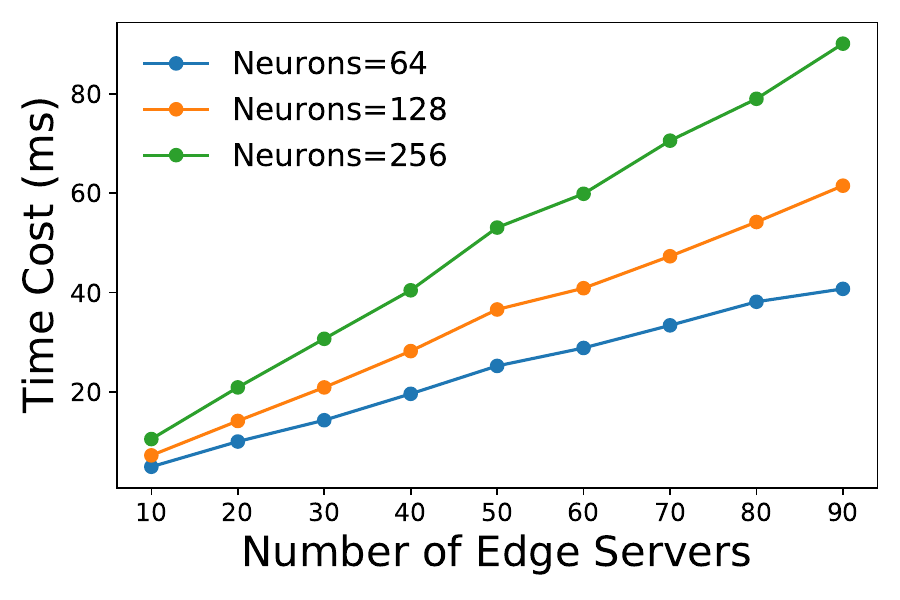}}
\subfigure[Different Clients' Data Distribution.]{
\label{fig_5_2}
\includegraphics[width=0.23\textwidth]{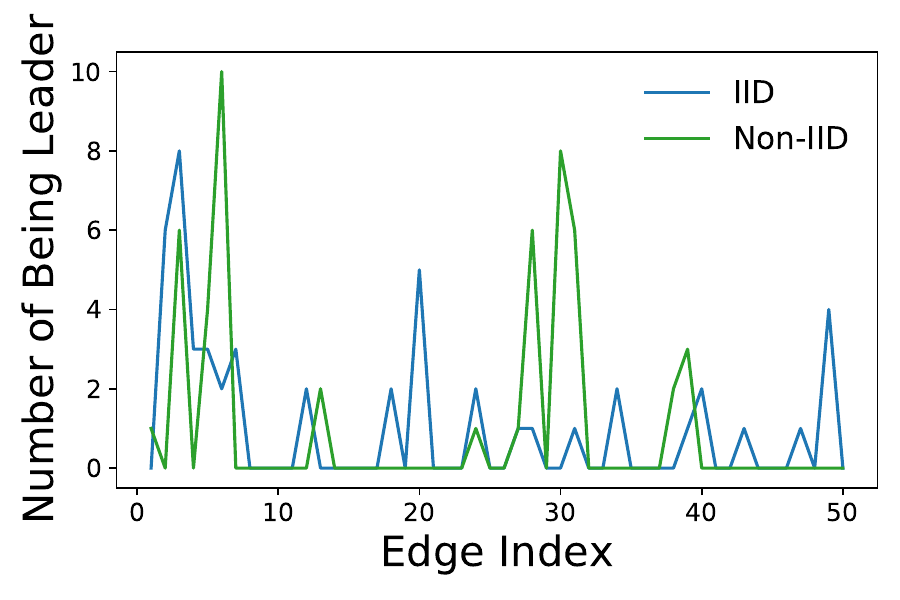}}
\caption{Computation Cost and Randomness of ME.}
\label{fig_5}
\end{figure}
\subsection{Evaluation of Bayesian Truth Serum-based Voting}
We validate the effectiveness of the BTSV by demonstrating its ability to distinguish malicious BCFL nodes from honest ones and limit their influences on voting results through the differences in the weight of votes (WV). Here we refer to honest and malicious BCFL nodes as honest nodes (HNs) and malicious nodes (MNs), respectively. For comparison, we design two attack strategies: targeted attack (TA) and random attack (RA), where TA means that all MNs collude with each other and vote the same BCFL node as the leader, while RA implies that each MN randomly casts a vote for a BCFL node. Also, the value of $G_{max}$ is set to 0.99 and $G_{min}$ is $\frac{1-0.99}{N-1}$ for model evaluation. We denote the chance of a malicious node behaving maliciously as CBM. And CHS considers 20 consecutive rounds of historical scores if available (i.e., $c$ is set to 20). Coefficients $\beta$, $\theta$, and $\epsilon$ used to derive WV are set to 1.3, 0.4, and 1.2, respectively. %\qh{why these values? will other options derive similar results? } 
As discussed in Section \ref{sec:BTS}, these coefficients limit the range of $WV^i(k)$ between 0 and 1.3, and ensure $WV^i(k) \approx 1$ when CHS is 0. Other coefficients derive similar results as the changing trend mainly depends on CHSs.% and resulting voting weights simply reflect the CHSs. 

By comparing the BTSV under these two attack strategies, we present the experimental results in Fig. \ref{fig_6}. The differences in the average weight of the vote of HNs and MNs are clear after several rounds of BTSV iterations. It can be seen that all the trends of voting weight under RA are steeper than that under TA. This means the difference between voting weights is more obvious under RA, while TA is always more difficult to be restricted than RA regardless of the proportion of MNs and the probability of MNs performing attacks. %\qh{how can readers derive this conclusion? what are the observations from figures?}. 
%It is worth noting that the total number of nodes (edge servers) will not affect trends presented in Fig. \ref{fig_6}. Since the major components ($\overline{x}_j$, $\overline{y}_j$ and $p_j^i(k)$) used in calculating BTS score are all variables related to the percentage of all voters' actions instead of the actual number of voters. The changes in the total number of nodes will not affect the CHSs calculation and thus have no effects on voting weights. %\qh{how can readers derive this conclusion given that you never mention the impact of the number of nodes before?}

We first discuss how the probability of MNs performing attacks (CBM) affects the difference in voting weights between HNs and MNs under both attack strategies. By comparing Figs. \ref{fig_6_1} and \ref{fig_6_2} or Figs. \ref{fig_6_3} and \ref{fig_6_4}, with the same proportion of MNs, more difference in voting weights is observed when CBM increases under both attacks. Thus it is easier to have the attack recognized and penalized when the probability of MNs performing attacks is high.

Then we investigate how the proportion of MNs impacts the difference in voting weights between HNs and MNs. By comparing Figs. \ref{fig_6_1} and \ref{fig_6_3} or Figs. \ref{fig_6_2} and \ref{fig_6_4}, with the same CBM, less difference in voting weights is observed when the proportion of MNs raises under TA, while the difference under RA remains similar.  
%\qh{compare which to which? it is hard to understand in this writing. can we only compare two figures each time? again, you forgot to describe the results in the figure.} 
This indicates that changes in the proportion of MNs do not affect the amount of penalty imposed on MNs under RA. While for TA, when either the proportion of MNs is low or CBM is high, more penalties on voting weights will be distributed to MNs.

Lastly, we analyze the corresponding reasons from the perspective of BTS score. This score represents how common the selected option is compared to the common prediction, and the prediction score measures the difference between the actual and predicted distribution of selections. In RA, a rise in the proportion of MNs will lead to the same increase in score components for MNs and HNs because additional MNs vote for random BCFL nodes instead of already cast votes. This introduces the same additional elements for summation in (\ref{xa}), (\ref{ya}), and (\ref{predictscore}), thus leading to the same increase in eventual scores and consequently a similar difference between average WV for both MNs and HNs. However, this is not the case for TA since a change in the probability of MNs performing attacks or the proportion of MNs will cause a direct value difference in the information and prediction score components (the number of votes for the same selection is changed). Such differences are amplified after deriving the weight of the vote using a sigmoid function.
\begin{figure}[ht]
\centering
%\subfigure[20\% Malicious Nodes, Each Has 30\% Chance Behaving Malignant]{
\subfigure[20\% MNs, CBM=30\%.]{
\label{fig_6_1}
\includegraphics[width=0.23\textwidth]{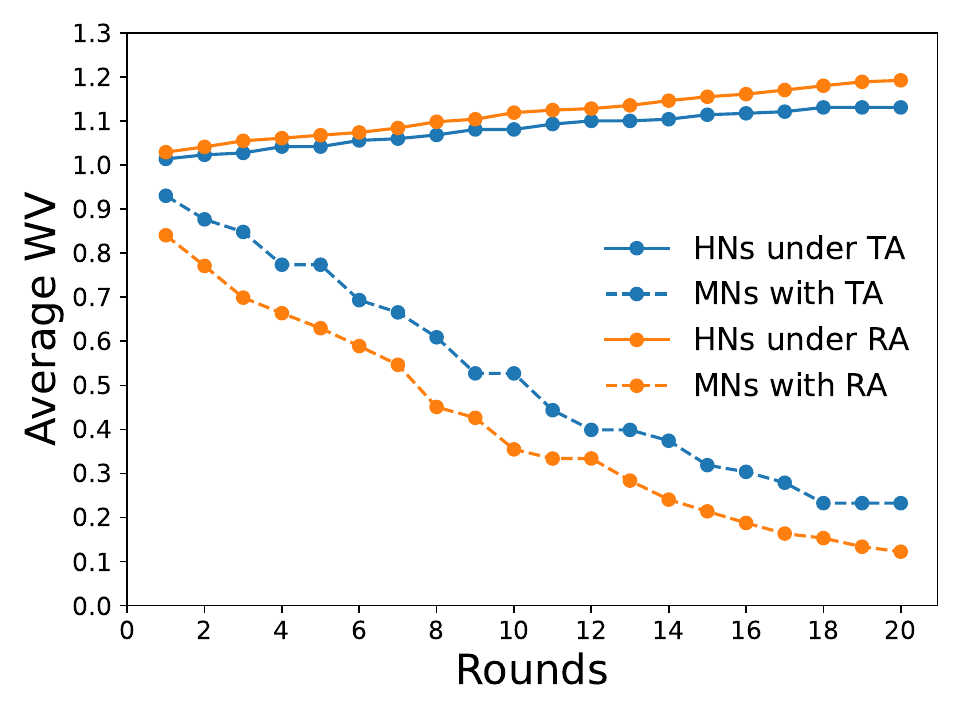}}
\subfigure[20\% MNs, CBM=70\%.]{
\label{fig_6_2}
\includegraphics[width=0.23\textwidth]{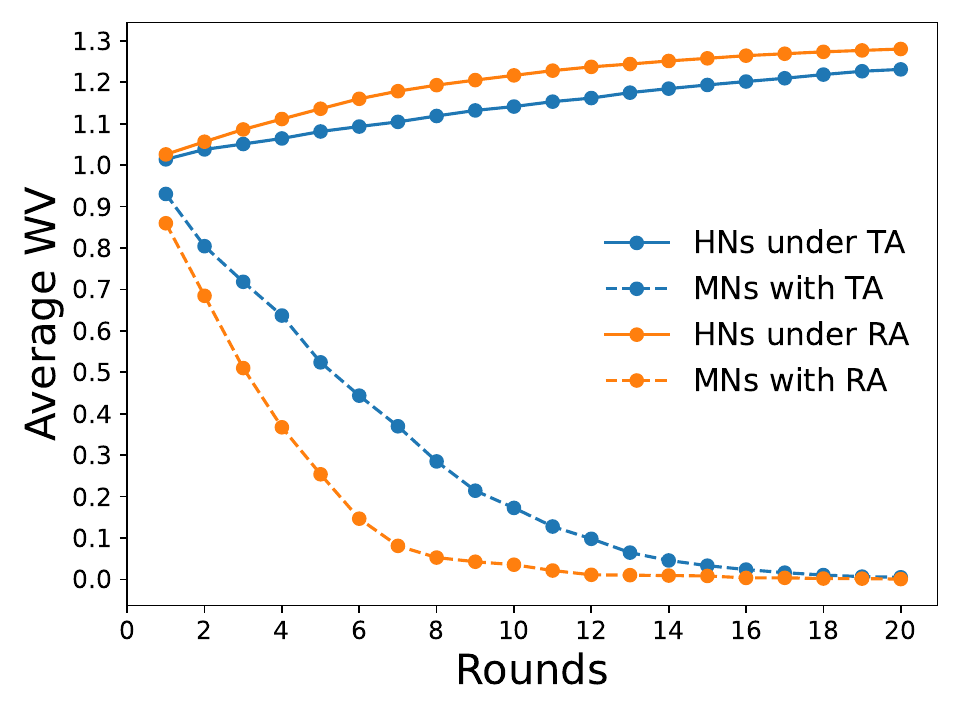}}
\subfigure[40\% MNs, CBM=30\%.]{
\label{fig_6_3}
\includegraphics[width=0.23\textwidth]{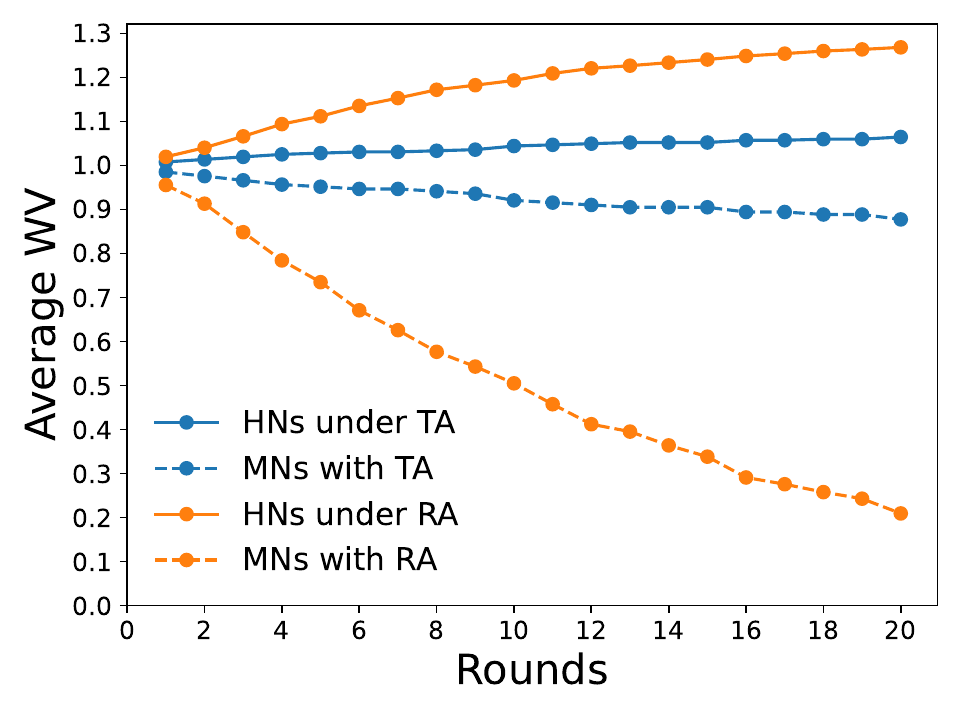}}
\subfigure[40\% MNs, CBM=70\%.]{
\label{fig_6_4}
\includegraphics[width=0.23\textwidth]{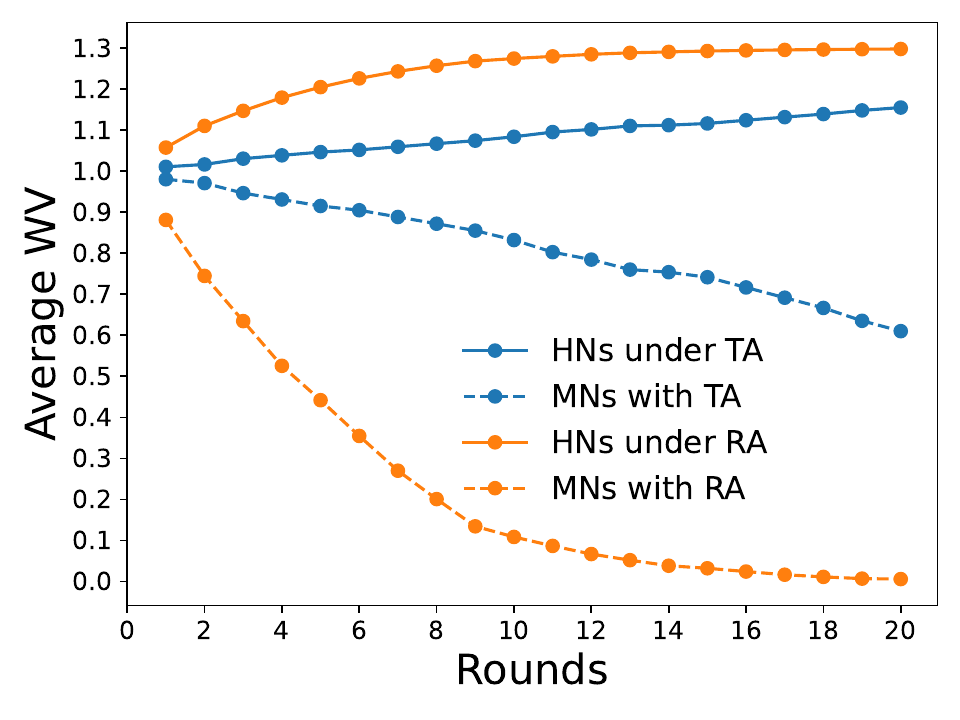}}
\caption{Average weight of the vote for Honest Nodes (HNs) and Malicious Nodes (MNs) Performing Targeted Attack (TA) vs. Random Attack (RA). }
\label{fig_6}
\end{figure}

\subsection{Evaluation of Incentive Mechanism}
Finally, we study the effects of strategies of the task publisher $\delta$ and BCFL nodes $f_i$ on utility functions. In this experiment, we set $B$ = 500, $\varphi$ = 5, $\lambda$ = 1, $\mu_i$ = 5, $\sum f_{-i}=1000$ and $\gamma_i$ = 0.01 to satisfy constraints of utility functions. 

First, impacts of the total rewards $\delta$ and total CPU cycle frequencies $F$ on the model owner's utility are evaluated. 
%It is worth noting that regarding the task publisher, the change in these parameters will only affect its optimal strategy $\delta^*$ but not its utility. Since the utility function we proposed for the task publisher describes an expectation between the optimal total rewards $\delta^*$ and total efforts (CPU cycle frequencies) $F$ provided by edge servers. It always aims to match the expectation, and thus utility is maximized to $B$ constantly. 
The total CPU cycle frequencies $F$ is set to 1000 when evaluating the strategy of model owner $\delta$. And the total rewards $\delta$ is set to 5000 when evaluating the sum of strategies of BCFL nodes $F$. As illustrated in Fig. \ref{fig_7_1}, the model owner's utility $U_{tp}$ is inversely proportional to $F$. This is as expected as the denominator of $U_{tp}$ is $F$. Fig. \ref{fig_7_3} demonstrates the existence of the optimal strategy $\delta$ for the model owner. 

Next, we evaluate the impacts of the total rewards $\delta$ and total CPU cycle frequency $f_i$ on a BCFL node $e_i$'s utility. The total CPU cycle frequency $f_i$ is set to 40 when evaluating the strategy of model owner $\delta$. The total rewards $\delta$ is set to 5000 when evaluating the strategy of BCFL node $e_i$. It can be seen in \eqref{Ui}, the BCFL node $e_i$’s utility $U_i$ and $\delta$ show a linear relationship, as indicated in Fig. \ref{fig_7_2}. Similarly, Fig. \ref{fig_7_4} demonstrates the existence of the optimal strategy $f_i$ for BCFL node $e_i$. 

\begin{figure}[ht]
\centering
\subfigure[$F$ vs. $U_{tp}$.]{
\label{fig_7_1}
\includegraphics[width=0.23\textwidth]{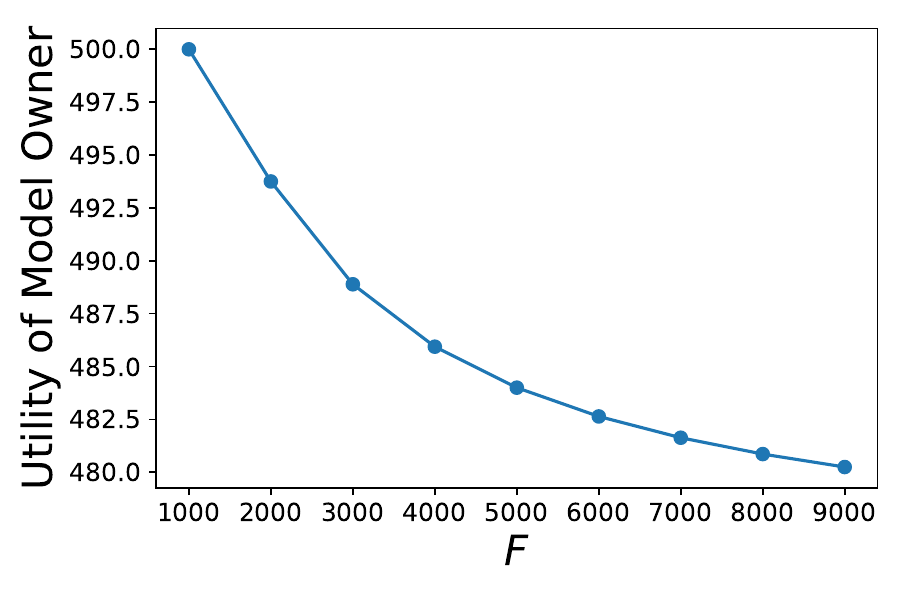}}
\subfigure[$\delta$ vs. $U_i$.]{
\label{fig_7_2}
\includegraphics[width=0.23\textwidth]{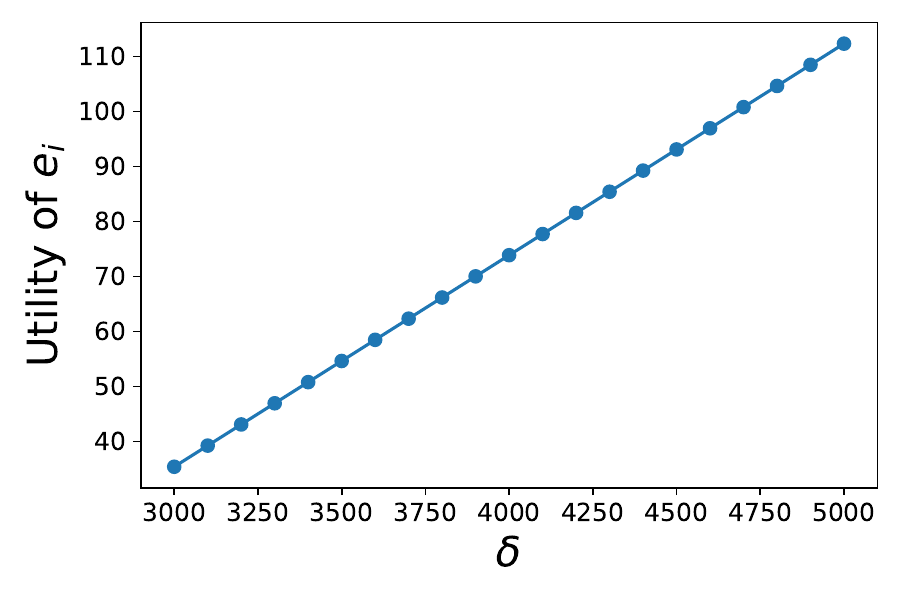}}
\subfigure[$\delta$ vs. $U_{tp}$. ]{
\label{fig_7_3}
\includegraphics[width=0.23\textwidth]{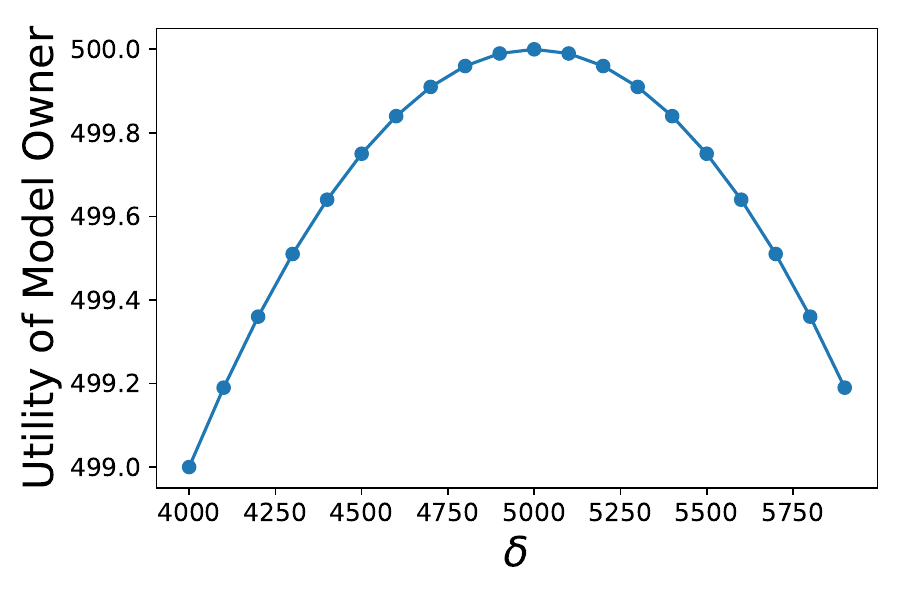}}
\subfigure[$f_i$ vs. $U_i$. ]{
\label{fig_7_4}
\includegraphics[width=0.23\textwidth]{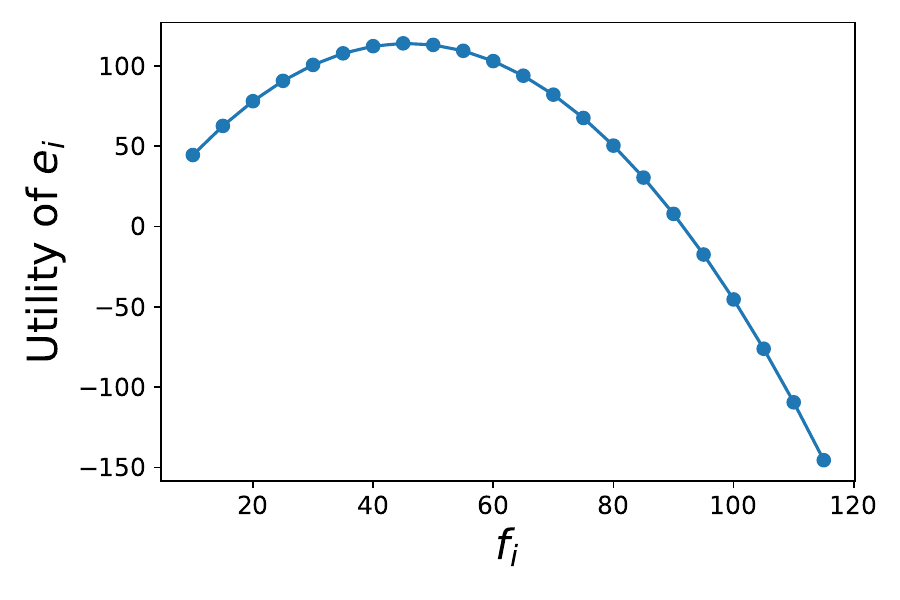}}
\caption{Impact of $F$ and $\delta$ on Utility Functions $U_{tp}$ and $U_i$.}
\label{fig_7}
\end{figure}

% \begin{figure}[ht]
% \centering
% \subfigure[Strategy of Model Owner $\delta$ vs. Utility of Owner $U_{tp}$ ]{
% \label{fig_8_1}
% \includegraphics[width=0.23\textwidth]{utility_mo_strategy.pdf}}
% \subfigure[Strategy of Edge Server $f_i$ vs. Utility of Edge Server $U_i$ ]{
% \label{fig_8_2}
% \includegraphics[width=0.23\textwidth]{utility_client_strategy.pdf}}
% \caption{Existence of Optimal Strategies for Utility Functions $U_i$ and $U_{tp}$  }
% \label{fig_8}
% \end{figure}

\section{Conclusion}
\label{sec:discussion}
In this paper, we propose a BHFL framework with a novel efficient consensus algorithm to avoid extra waste of computation resources during blockchain consensus. The consensus algorithm utilizes intermediate models from BCFL nodes and compares each model with the updated global model for leader node election. Additionally, we combine the Hash-based Commitment and Digital Signature Algorithm to resolve the model plagiarism issue that occurs during the model exchange process; and we also employ BTSV in the vote tally smart contract for determining the leader node so as to prevent voting bribery from malicious nodes. A two-stage Stackelberg game is modeled to design the incentive mechanism, providing motivation for clients in the FEL clusters to contribute to the learning task. Experimental results demonstrate the validity, fairness, and efficiency of our proposed system and mechanisms. In the future, we plan to design a consensus mechanism that enables more fairness among blockchain nodes regarding leader election when the training data is non-IID.

\bibliographystyle{IEEEtran}
\bibliography{refs}
\begin{IEEEbiography}[{\includegraphics[width=1in,height=1.25in,clip,keepaspectratio]{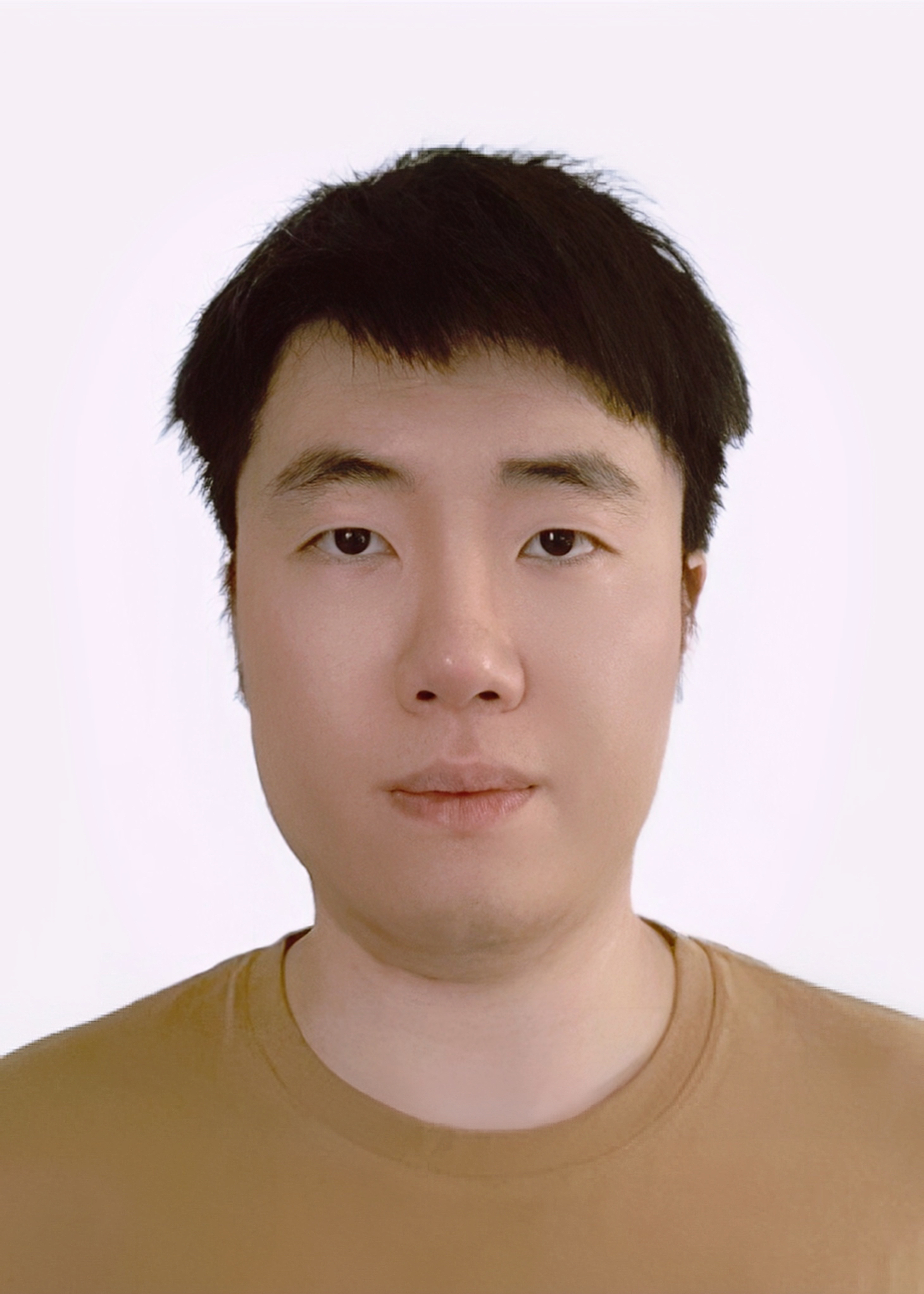}}]{Shengyang Li} 
received the M.Eng. degree in Computer Science and Electronics from the University of Bristol in 2019. He is currently working toward the Ph.D. degree in Electrical and Computer Engineering at Indiana University-Purdue University Indianapolis (IUPUI). He is a Research Assistant
with IUPUI.
\end{IEEEbiography}
\begin{IEEEbiography}[{\includegraphics[width=1in,height=1.25in,clip,keepaspectratio]{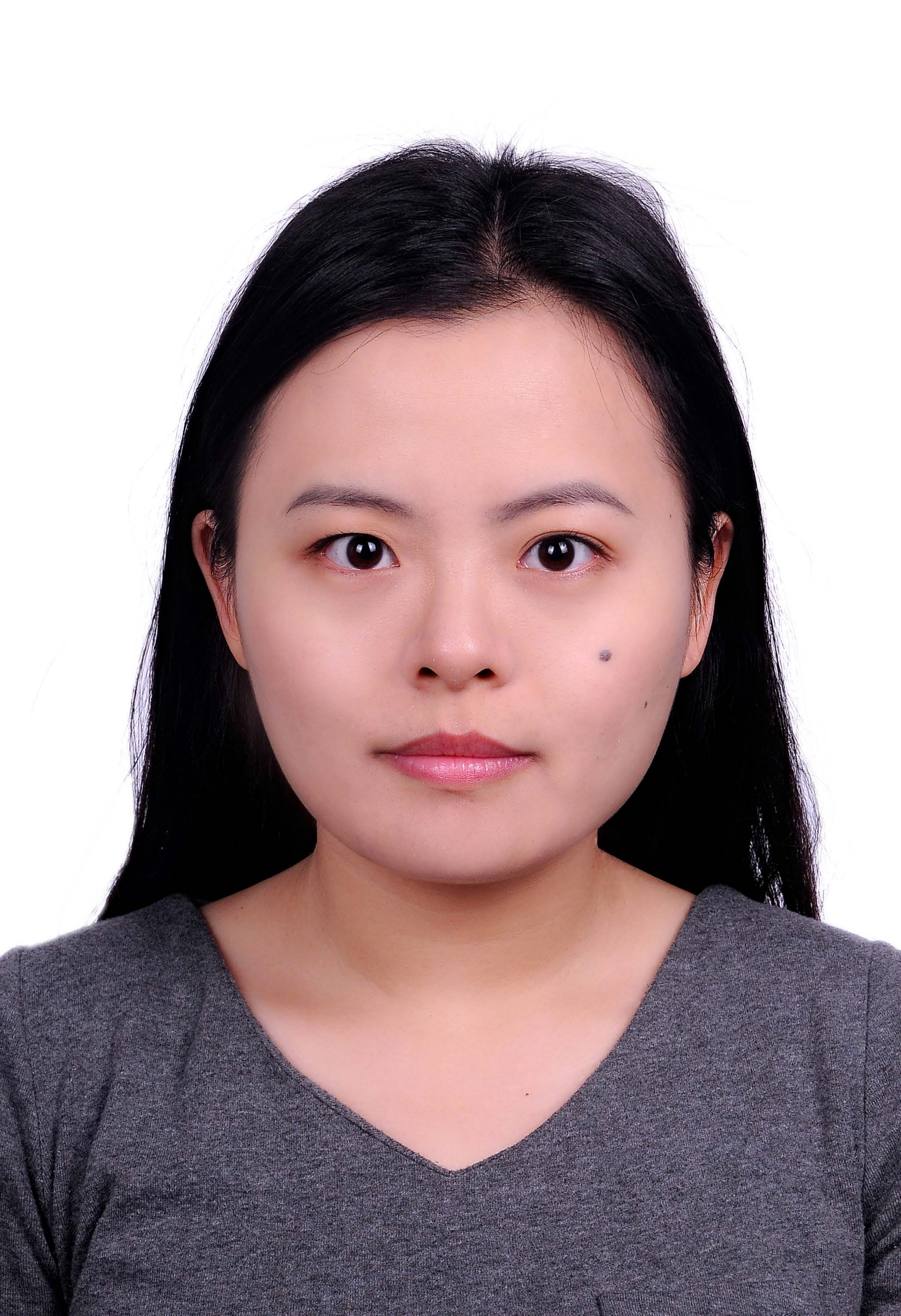}}]{Qin Hu} received her Ph.D. degree in Computer Science from the George Washington University in 2019. She is currently an Assistant Professor with the Department of Computer and Information Science, Indiana University-Purdue University Indianapolis (IUPUI). She has served on the Editorial Board of two journals, the Guest Editor for multiple journals, the TPC/Publicity Co-chair for several workshops, and the TPC Member for several international conferences. Her research interests include wireless and mobile security, edge computing, blockchain, and federated learning.
\end{IEEEbiography}

\begin{IEEEbiography}[{\includegraphics[width=1in,height=1.25in,clip,keepaspectratio]{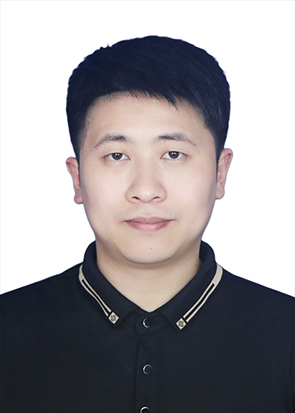}}]{Zhilin Wang} received his B.S. from Nanchang University in 2020. He is currently pursuing his Ph.D. degree in Computer and Information Science at Indiana University-Purdue University Indianapolis (IUPUI). He is a Research Assistant with IUPUI, and he is the reviewer of IEEE TPDS, IEEE IoTJ, Elsevier JNCA, IEEE TCCN, IEEE ICC'22, IEEE Access, and Elsevier HCC; he also serves as the TPC member of the IEEE ICC'22 Workshop. His research interests include blockchain, federated learning, edge computing, and optimization theory.
\end{IEEEbiography}

\end{document}